\documentclass[a4paper,UKenglish
,cleveref%
, autoref%
]{lipics-v2019}

\hideLIPIcs
\nolinenumbers

\DeclareFontFamily{U}{mathb}{\hyphenchar\font45}
\DeclareFontShape{U}{mathb}{m}{n}{
      <5> <6> <7> <8> <9> <10> gen * mathb
      <10.95> mathb10 <12> <14.4> <17.28> <20.74> <24.88> mathb12
      }{}
\DeclareSymbolFont{mathb}{U}{mathb}{m}{n}

\DeclareMathSymbol{\sqsubsetneq}{3}{mathb}{"88}
\usepackage{thm-restate}

\usepackage{graphicx}

\usepackage{pgf}
\usepackage{tikz}
\usepackage{tikz-cd}
\usetikzlibrary{matrix,trees,arrows,fit, positioning}
\usetikzlibrary{decorations.pathmorphing,patterns}
\usetikzlibrary{calc,patterns,decorations.markings}
\usetikzlibrary{backgrounds}

\tikzset{
  showback/.style={
    background rectangle/.style={
      draw=black!15,
      fill=blue!2,
      rounded corners=1ex}
  }
}

\tikzcdset{
  conflict/.style={
    -,
    dotted
  },
  ac/.style={
     dotted
  },
  flow/.style={
     ->>
  },   
  boxedcd/.style={
    every matrix/.append style={
      draw=black!50,
      fill=blue!2,
      very thick,
      rounded corners,
      inner sep=3pt
      #1
    }
  },
  poset/.style={
    every matrix/.append style={
      draw=black!50,
      fill=blue!2,
      thick,
      rounded corners,
      inner sep=0.8ex,
      font = \small,
      #1
    },
    every cell/.style = {
      inner sep=0.3ex,
    }
  }
}

\graphicspath{{fig/}}

\newcommand{\cons}{\ensuremath{\smallfrown}}
\newcommand{\setcons}[1]{\ensuremath{\smash{\raisebox{0.25em}{$\smallfrown$}} {#1}}}

\newcommand{\ac}{\ensuremath{\nearrow}}
\newcommand{\scauses}[1]{\lfloor #1 )}
\newcommand{\causes}[1]{\lfloor #1 \rfloor}
\newcommand{\consequences}[1]{\lceil #1 \rceil}
\newcommand{\cset}[1]{\mathit{conc}({#1})}

\newcommand{\cat}[1]{\ensuremath{\mathbf{#1}}}

\newcommand{\hist}[2]{{#1}[{#2}]}
\newcommand{\Hist}[1]{\ensuremath{\mathit{Hist}({#1})}}

\newcommand{\hset}[1]{\ensuremath{\mathit{hs}({#1})}}
\newcommand{\flt}[1]{\ensuremath{\mathit{fl}({#1})}}

\newcommand{\eqclass}[2]{\ensuremath{[{#1}]_{\scriptscriptstyle {#2}}}}

\newcommand{\quotient}[2]{\ensuremath{{#1}_{/\scriptscriptstyle {#2}}}}

\newcommand{\foldeq}[1]{\ensuremath{\mathit{FEq}({#1})}}

\newcommand{\conf}[1]{\ensuremath{\mathit{Conf}({#1})}}

\newcommand{\trans}[1]{\xrightarrow{#1}}

\newcommand{\prc}{\ensuremath{\mathrel{\raisebox{1.5pt}{$\curvearrowright$}}}}

\newcommand{\aes}{\textsc{aes}}
\newcommand{\fes}{\textsc{fes}}
\newcommand{\pes}{\textsc{pes}}

\newcommand{\pr}{\ensuremath{\mathbb{P}}}

\newcommand{\ES}[1][E]{\ensuremath{\mathsf{#1}}}

\newcommand{\Powfin}[1]{\ensuremath{\mathbf{2}_\mathit{fin}^{#1}}}

\pgfdeclarelayer{background}
\pgfdeclarelayer{foreground}
\pgfsetlayers{background,main,foreground}

\title{Minimisation of Event Structures}

\author{Paolo Baldan}{
  University of Padova, Italy
}
{baldan@math.unipd.it}{https://orcid.org/0000-0001-9357-5599}{}%

\author{Alessandra Raffaet\`a%
}{University Ca' Foscari of Venice, Italy}{raffaeta@unive.it}{[orcid]}{}%

\authorrunning{P. Baldan and A.Raffaet\`a}%

\Copyright{Paolo Baldan and Alessandra Raffaet\`a}%



\category{}%

\relatedversion{}%

\supplement{}%

\newtheorem{atheorem}{Theorem}[section]
\newtheorem{alemma}[atheorem]{Lemma}
\newtheorem{adefinition}[atheorem]{Definition}
\newtheorem{aproposition}[atheorem]{Proposition}

\begin{document}

\maketitle

\begin{abstract}
  Event structures are fundamental models in concurrency theory,
  providing a representation of events in computation and of their
  relations, notably
  concurrency, conflict and causality.
  In this paper we present a theory of minimisation for event
  structures. Working in a class of event structures that generalises many
  stable event structure models in the literature (e.g., prime,
  asymmetric, flow and bundle event structures), we study a notion of
  behaviour-preserving quotient, referred to as a folding, taking
  (hereditary) history preserving bisimilarity as a reference
  behavioural equivalence.
  We show that for any event structure a folding producing a uniquely
  determined minimal quotient always exists.  We observe that each
  event structure can be seen as the folding of a prime event
  structure, and that all foldings between general event structures
  arise from foldings of (suitably defined) corresponding prime event
  structures. This gives a special relevance to foldings in the class
  of prime event structures, which are studied in detail. We identify
  folding conditions for prime and asymmetric event structures, and show
  that also prime event structures always admit a unique minimal
  quotient (while this is not the case for various other event
  structure models).
\end{abstract}

\section{Introduction}

When dealing with formal models of computational systems, a classical
problem is that of minimisation, i.e., for a given system, define and
possibly construct a compact version of the system which, very roughly
speaking, exhibits the same behaviour as the original one, avoiding
unnecessary duplications.
The minimisation procedure depends on the notion of behaviour of
interest and also on the expressive power of the formalism at hand,
which determines its capability of describing succinctly some
behaviour.
One of the most classical examples is that of finite state automata,
where one is typically interested in the accepted language. Given a
deterministic finite state automaton, a uniquely determined minimal
automaton accepting the same language can be constructed, e.g., as a
quotient of the original automaton via a partition/refinement
algorithm (see, e.g.,~\cite{HMU:IAT}). Moving to non-deterministic
finite automata,
minimal automata become smaller, at the price of a
computationally more expensive minimisation procedure and
non-uniqueness of the minimal automaton~\cite{MS:EPRE}.

In this paper we study the problem of minimisation for event
structures, a fundamental model in concurrency theory~\cite{Win:ES,Win:ECS}.
Event structures are a natural semantic model when
one is interested in modelling the dynamics of a system 
by providing an explicit representation of the events in computations (occurrence of atomic actions) and of the
relations between events, like causal dependencies, choices,
possibility of parallel execution, i.e., in what is referred to as a true concurrent (non-interleaving) semantics.
Prime event
structures~\cite{NPW:PNES}, probably the most widely used event
structure model, capture dependencies between events in terms of
causality and conflict. A number of variations of prime event structures have been introduced in the literature. In this paper we will deal with asymmetric event
structures~\cite{BCM:CNAED}, which generalise prime event structures
with an asymmetric form of conflict which allows one to model
concurrent readings and precedences between actions, and
flow~\cite{BC:SCCS,Bou:FESFN} and bundle~\cite{Lan92:BES} event structures,
which add the possibility of directly modelling disjunctive causes.
Event structures have been used for defining a concurrent semantics of
several formalisms, like Petri nets~\cite{NPW:PNES}, graph rewriting
systems~\cite{Handbook,Bal:PhD,Sch:RRSG} and process
calculi~\cite{Win:ESSCCS,VY:TESLP,BMM:ESSNC}.  Recent applications are
in the field of weak memory models~\cite{PS:CSRA,JR:OTRES,CV:GTARES}
and of process mining and differencing~\cite{DG:PMRES}.

Behavioural equivalences, defined in a true concurrent setting, take into
account not only the possibility of performing steps, but also the way
in which such steps relate with each other.
We will focus on hereditary history preserving
(hhp-)bisimilarity~\cite{Ber:HHP}, the finest equivalence in the true
concurrent spectrum in~\cite{vGG:ENCC-AI}, which, via the concept of
open map, has been shown to arise as a canonical behavioural
equivalence when considering partially ordered computations as
observations~\cite{JNW:BFOM}.

The motivation for the present paper originally stems from some work
on the analysis and comparison of business process models. The idea,
advocated in~\cite{DG:PMRES,ABDG:DBDBPMES}, is to use
event structures as a foundation for representing, analysing and
comparing process models. The processes, in their graphical
presentation, should be understandable, as much as possible, by a
human user, who should be able, e.g., to interpret the differences
between two processes diagnosed by a comparison tool.
For this aim it is important to avoid ``redundancies'' in the
representation and thus to reduce the number of events, but clearly
without altering the behaviour.  The paper~\cite{ABG:RESHPB} explores
the use
of asymmetric and flow event structures and, for such models, it introduces some ad-hoc reduction techniques that
allow one to merge sets of events without changing the true concurrent
behaviour.
A general notion of behaviour preserving quotient, referred to as a folding, is introduced over an abstract class of event structures, having asymmetric and flow event structures as subclasses. However, no general theory  is developed. The paper focuses on a special class of foldings, the so-called elementary foldings, which can only merge a single set of events into one event, and these are studied separately on each specific subclass of event structures (asymmetric and flow event structures), providing only sufficient conditions ensuring that a function is a folding.

A general theory of behaviour preserving quotients for event structures is  thus called for, settling some natural foundational questions. Is the notion of folding adequate, i.e., are all behaviour preserving quotients expressible in terms of foldings? Is there a minimal quotient in some suitably defined general class of event
structures? What happens in specific subclasses? (for asymmetric and flow event structures the answer is known to be negative, but for prime event structures the question is open). Working in the specific subclasses of  event structures, can we have a characterisation of general foldings, providing not only sufficient but also necessary conditions?

In this paper we address the above questions. We work in a general class of event structures based on
the idea of family of posets in~\cite{Ren:PC},
sufficiently expressive to generalise most
stable event structures models in the literature, including prime~\cite{NPW:PNES}, asymmetric~\cite{BCM:CNAED}, flow~\cite{BC:SCCS} and bundle~\cite{Lan92:BES} event structures.

As a first step we study, in this general setting, the notion of
folding, i.e., of behaviour preserving quotient.
A folding is a surjective
function that identifies some events while keeping the behaviour
unchanged.  Formally, it establishes a hhp-bisimilarity between the
source and target event structure.
Foldings can be characterised as open maps in the sense
of~\cite{JNW:BFOM}.
Actually, it turns out that not all behaviour preserving quotients
arise as a folding, but we show that for any behaviour preserving
quotient, there is a folding that induces a coarser equivalence, in a
way that foldings properly capture all possible behaviour preserving
quotients.
Additionally, given two possible foldings of an event structure we
show that it is always possible to ``join'' them. This allows to prove
that for each event structure a maximally folded version, namely a
uniquely determined minimal quotient always exists.

Relying on the order-theoretic properties of the set of
configurations of event structures~\cite{Ren:PC}, and on the
correspondence between prime event structures and
domains~\cite{NPW:PNES}, we derive that each event structure in the
considered class arises as the folding of a canonical prime event
structure.
Moreover, all
foldings between general event structures arise from foldings of the
corresponding canonical prime event structures.
Interestingly, this result can be derived from the
characterisation of folding morphisms as open maps.

The results above give a special relevance to foldings in the class of
prime event structures, which thus are studied in detail. We provide
necessary and sufficient conditions characterising foldings for prime event structures.
This allows establish a clear connection with the so-called
abstraction morphisms, introduced in~\cite{Cas:phd} for similar purposes.
This characterisation of foldings can guide, at least in the case of finite
structures, the construction of behaviour preserving
quotients. Moreover we show that also prime event structures always
admit a minimal quotient.

The fact that all event structures arise as foldings of prime event
structures allows one to think of various brands of event structures
in the literature, like asymmetric, flow, bundle event structures as
more expressive models that allow for smaller
realisations of a given behaviour, i.e., of smaller quotients. For all
these classes, however, the uniqueness of the minimal
quotient is lost.
Despite the fact that foldings on wider classes of event structures
can be studied on the corresponding canonical prime event structures,
a direct approach can be theoretically interesting and it can lead
more efficient minimisation procedures. In this paper, a characterisation of foldings is explicitly devised for asymmetric event structures.

Most results have a natural categorical interpretation, which is only
hinted at in the paper. In order to keep the presentation simple, the
categorical references are inserted in side remarks that can be safely
skipped by the non-interested reader. This applies, in particular, to
the possibility of viewing foldings as open maps in the sense
of~\cite{JNW:BFOM}.
This correspondence, which in the present paper only surfaces,
suggests the possibility of understanding and developing our results
in a more abstract categorical setting.
More details about this are provided in the appendices.

The rest of the paper is structured as follows. In
\S~\ref{se:preliminaries} we introduce the class of event structures
we work with, hereditary history preserving bisimilarity and we discuss
how various event structure models in the literature embed into the
considered class.
In \S~\ref{se:folding} we introduce and study the notion of folding, we prove the  existence of a minimal quotient and we show the tight relation between general foldings and those on prime event structures.
In \S~\ref{se:folding-criteria} we present folding criteria on prime and asymmetric event structures, and discuss the existence of minimal quotients.
Finally, in \S~\ref{se:conclusions} we draw some conclusions, discuss
connections with related literature and outline future work venues.
An appendix contains all proofs and some additional technical results.

\section{Event Structures and History Preserving Bisimilarity}
\label{se:preliminaries}

In this section we define \emph{hereditary history-preserving
  bisimilarity}, the reference behavioural equivalence in the
paper. This is done for an abstract notion of event structure,
introduced in~\cite{Ren:PC}, in a way that various stable event
structure models in the literature can be seen as special
subclasses. We will explicitly discuss prime~\cite{NPW:PNES},
asymmetric~\cite{BCM:CNAED}, flow~\cite{BC:SCCS,Bou:FESFN} and
bundle~\cite{Lan92:BES} event structures.

\subparagraph{Notation.}  We first fix some basic notation on sets,
relations and functions. Let $\mathrel{r} \subseteq X \times X$ be a
binary relation.
Given $Y, Z \subseteq X$,  we write
$Y \mathrel{r}^\forall Z$ (resp. $Y \mathrel{r}^\exists Z$) if for all
(resp. for some) $y \in Y$ and $z \in Z$ it holds $y \mathrel{r} z$. When $Y$ or $Z$ are singletons, sometimes we replace them by their only element, writing, e.g., $y \mathrel{r}^\exists Z$ for $\{y\} \mathrel{r}^\exists Z$.
The
relation $r$ is \emph{acyclic} on $Y$ if
there is no $\{ y_0, y_1, \ldots, y_n \} \subseteq Y$ such that $y_0 \mathrel{r}  y_1 \mathrel{r} \ldots r\ y_n \mathrel{r}  y_0$.
Relation $r$ is a \emph{partial order} if it is reflexive,
antisymmetric and transitive.
Given a function
$f: X \to Y$ we will denote by
$f[x \mapsto y] : X \cup \{ x \} \to Y \cup \{ y \}$ the function
defined by $f[x \mapsto y] (x) = y$ and $f[x \mapsto y] (z) = f(z)$
for $z \in X \setminus \{ x \}$.
Note that the same notation can represent an update of $f$, when $x
\in X$, or an extension of its domain, otherwise. For $Z \subseteq X$, we denote by $f_{|Z} : Z \to Y$ the restriction of $f$ to $Z$.

\subsection{Event Structures}

Following~\cite{Ren:PC,vGla:HPPG,GP:CSESPN,ABG:RESHPB}, we work on a class of
event structures where configurations are given as a primitive
notion. More precisely, we borrow the idea of family of posets
from~\cite{Ren:PC}. 

\begin{definition}[family of posets]
  \label{de:poset}
  A \emph{poset} is a pair $(C, \leq_C)$ where $C$ is a set and
  $\leq_C$ is a partial order on $C$. A poset will be often denoted
  simply as $C$, leaving the partial order relation $\leq_C$ implicit.
  Given two posets $C_1$ and $C_2$ we say that $C_1$ is a
  \emph{prefix} of $C_2$ and write $C_1 \sqsubseteq C_2$ if
  $C_1 \subseteq C_2$ and
  for all $x_1 \in C_1$, $x_2 \in C_2$, if $x_2 \leq_{C_2} x_1$ then
  $x_2 \in C_1$ and $x_2 \leq_{C_1} x_1$.
  A \emph{family of posets} $F$ is a prefix-closed set of finite
  posets i.e., a set of finite posets such that if $C_2 \in F$ and
  $C_1 \sqsubseteq C_2$ then $C_1 \in F$.  We say that two posets
  $C_1, C_2 \in F$ are compatible, written $C_1 \cons C_2$, if they
  have an upper bound, i.e., there is $C \in F$ such that
  $C_1, C_2 \sqsubseteq C$. The family of posets $F$ is called
  \emph{coherent} if each subset of $F$ whose elements are pairwise
  compatible has an upper bound.
\end{definition}

Posets $C$ will be used to represent configurations, i.e., sets of
events executed in a computation of an event structure. The order
$\leq_C$ intuitively represents the order in which the events in $C$
can occur. This motivates the prefix order that can be read as a
computational extension: when $C_1 \sqsubseteq C_2$ we have that
$C_1 \subseteq C_2$, with events in $C_1$ ordered exactly as in $C_2$,
and the new events in $C_2 \setminus C_1$ cannot precede events
already in $C_1$. An example of family of posets can be found
in~Fig.~\ref{fi:es} (left). Observe, for instance, that the
configuration with set of events $\{ c \}$ is not a prefix of the one
with set of events $\{a,c\}$, since in the latter
$a \leq c$.

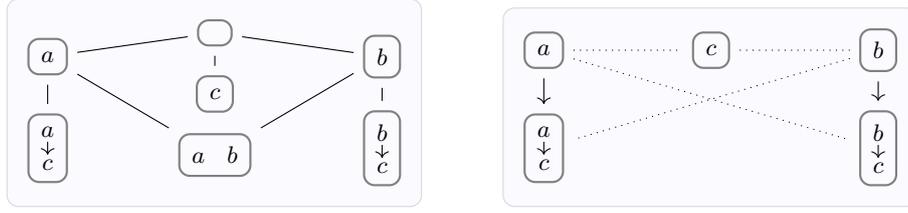
\begin{figure}
  \hfill
\begin{tikzpicture}[x=22mm,y=9mm, show background rectangle, showback]
  \node at (1,3) (empty) {
    \begin{tikzcd}[poset]
      \ 
    \end{tikzcd}
  };
  
  \node at (0,2.65) (a) {
    \begin{tikzcd}[poset]
      a
    \end{tikzcd}
  };
  
  \node at (1,2.1) (c) {
    \begin{tikzcd}[poset]
      c
    \end{tikzcd}
  };

  \node at (2,2.65) (b) {
    \begin{tikzcd}[poset]
      b
    \end{tikzcd}
  };

  \node at (0,1.3) (ac) {
    \begin{tikzcd}[poset, sep=2mm]
      a \arrow[d]\\
      c
    \end{tikzcd}
  };

  \node at (1,1.2) (ab) {
    \begin{tikzcd}[poset, sep=2mm]
      a &  b      
    \end{tikzcd}
  };

  \node at (2,1.3) (bc) {
    \begin{tikzcd}[poset, sep=2mm]
      b \arrow[d]\\
      c
    \end{tikzcd}
  };
  
  \draw (empty) -- (a);
  \draw (empty) -- (b);
  \draw (empty) -- (c);
  \draw (a)     -- (ab);
  \draw (b)     -- (ab);
  \draw (a)     -- (ac);
  \draw (b)     -- (bc);
\end{tikzpicture}
\hfill
\begin{tikzpicture}[x=22mm,y=13mm, show background rectangle, showback]
  \node at (0,2) (a) {
    \begin{tikzcd}[poset]
      a
    \end{tikzcd}
  };
  
  \node at (1,2) (c) {
    \begin{tikzcd}[poset]
      c
    \end{tikzcd}
  };

  \node at (2,2) (b) {
    \begin{tikzcd}[poset]
      b
    \end{tikzcd}
  };

  \node at (0,1) (ac) {
    \begin{tikzcd}[poset, sep=2mm]
      a \arrow[d]\\
      c
    \end{tikzcd}
  };

  \node at (2,1) (bc) {
    \begin{tikzcd}[poset, sep=2mm]
      b \arrow[d]\\
      c
    \end{tikzcd}
  };
  
  \draw [->]     (a)     -- (ac);
  \draw [->]     (b)     -- (bc);
  \draw [dotted] (a)     -- (c);
  \draw [dotted] (b)     -- (c);
  \draw [dotted] (a)     -- (bc);
  \draw [dotted] (b)     -- (ac);
\end{tikzpicture}
\hfill\mbox{}

\caption{An event structure $\ES$ and the canonical {\pes} $\pr(\ES)$}
\label{fi:es}
\end{figure}

An event structure is then defined simply as a coherent family of
posets where events carry a label. Hereafter $\Lambda$ denotes a fixed
set of labels.

\begin{definition}[event structure]
  \label{de:event-structure}
  A \emph{(poset) event structure} is a tuple
  $\ES = \langle E, \conf{\ES}, \lambda \rangle$ where $E$ is a set of
  events, $\conf{\ES}$ is a coherent family of posets such that
  $E = \bigcup \conf{\ES}$ and $\lambda : E \to \Lambda$ is a
  labelling function.
  For a configuration $C \in \conf{\ES}$ the order $\leq_C$ is
  referred to as the \emph{local order}.
\end{definition}
In~\cite{ABG:RESHPB} abstract event structures are defined as a collection of ordered configurations, without any further constraint. This is sufficient for giving some general definitions which are then studied in specific subclasses of event structures. Here, in order to develop a theory of foldings at the level of general event structures, we need to assume stronger properties, those of a family of posets from~\cite{Ren:PC} (e.g, the fact that Definition~\ref{de:pes-es} is well-given relies on this).
This motivates the name poset event structure.
Also note that, differently from what happens in other general concurrency models, like configuration structures~\cite{GP:CSESPN}, configurations are endowed explicitly with a partial order, which in turn intervenes in the definition of the prefix order between configurations. This will be essential to view asymmetric or flow event structures as subclasses.
Since we only deal with poset event structures and
their subclasses, we will often omit the qualification ``poset'' and refer
to them just as event structures.
Moreover, we will often identify an event structure $\ES$ with the
underlying set $E$ of events and write, e.g., $x \in \ES$ for
$x \in E$.

An \emph{isomorphism of configurations} $f : C \to C'$ is an
isomorphism of posets that respects the labelling, namely for all
$x, y \in C$, we have $\lambda(x) = \lambda(f(x))$ and $x \leq_{C} y$
iff $f(x) \leq_{C'} f(y)$. When configurations $C, C'$ are isomorphic
we write $C \simeq C'$.
As mentioned above, the prefix order on configurations can be
interpreted as computational extension. This will be later formalised
by a notion of transition system over the set of configurations (see~\cref{de:transition}).

Given an event $x$ in a configuration $C$ it will be useful to refer
to the prefix of $C$ including only those events that
necessarily precede $x$ in $C$ (and $x$ itself). This motivates
  the following definition.

\begin{definition}[history]
  Let $\ES$ be an event structure, let $C \in \conf{\ES}$ and let
  $x \in C$. The history of $x$ in $C$ is defined as the set
  $\hist{C}{x} = \{ y \in C \mid y \leq_C x \}$ endowed with the
  restriction of $\leq_C$ to $\hist{C}{x}$, i.e.,
  $\leq_{\hist{C}{x}} = \leq_C \cap (\hist{C}{x} \times
  \hist{C}{x})$. The set of histories in $\ES$ is
  $\Hist{\ES} = \{ \hist{C}{x} \mid C \in \conf{\ES}\ \land\ x\in C
  \}$. The set of histories of a specific event $x \in \ES$ will be
  denoted by $\Hist{x}$.
\end{definition}

\subsection{Hereditary History Preserving Bisimilarity}

In order to define history preserving bisimilarity, it is
convenient to have an explicit representation of the transitions
between configurations.

\begin{definition}[transition system]
  \label{de:transition}
  Let $\ES$ be an event structure. If $C, C' \in \conf{\ES}$
  with $C \sqsubseteq C'$ we write $C \trans{X} C'$ where
  $X = C' \setminus C$.
\end{definition}
When $X$ is a singleton, i.e., $X = \{ x \}$, we will often write
$C \trans{x} C'$ instead of $C \trans{\{x\}} C'$.
It is easy to see that in an event structure each configuration is reachable in the transition system from the empty one.

As it happens in the interleaving approach, a bisimulation between two
event structures requires any event of an event structure to be
simulated by an event of the other, with the same label.
Additionally, the two events are required to have the same ``causal
history''.

\begin{definition}[(hereditary) history preserving bisimilarity]
  \label{de:hp-bisim}
  Let $\ES$, $\ES'$ be event structures. A \emph{history preserving
    (hp-)bisimulation} is a set $R$ of triples $(C, f, C')$, where
  $C \in \conf{\ES}$, $C' \in \conf{\ES'}$ and $f : C \to C'$ is an
  isomorphism of configurations, such that
  $(\emptyset,\emptyset,\emptyset) \in R$ and for all
  $(C_1, f, C_1')\in R$
  \begin{enumerate}
  \item for all $C_1 \trans{x} C_2$ there exists some
    $C_1' \trans{x'} C_2'$ such that
    $(C_2, f[x \mapsto x'], C_2') \in R$;
      
  \item for all $C_1' \trans{x'} C_2'$ there exists some
    $C_1 \trans{x} C_2$ such that
    $(C_2, f[x \mapsto x'], C_2') \in R$.
  \end{enumerate}
  Relation $R$ is called a \emph{hereditary history preserving
    (hhp-)bisimulation} if, in addition, it is downward closed, i.e., if $(C_1, f, C_1')\in R$ and $C_2 \subseteq C_1$ then $(C_2, f_{|C_2}, f(C_2)) \in R$.
\end{definition}
Observe that, in the definition above, an event must be simulated by
an event with the same label. In fact, in the triple
$(C \cup \{x\}, f[x \mapsto x'], C' \cup \{x'\}) \in R$, the second
component $f[x \mapsto x']$ must be an isomorphism of configurations,
i.e., of labelled posets, and thus it preserves labels.
Hhp-bisimilarity has been shown to arise as a canonical behavioural
equivalence on prime event structures, as an instance of a general notion
defined in terms of the concept of open map, when considering
partially ordered computations as observations~\cite{JNW:BFOM}.

\subsection{Examples: Prime, Asymmetric, Flow and  Bundle Event Structures}

We next observe how different kinds of event structures, introduced for
various purposes in the literature, can be naturally viewed as subclasses of the
poset event structures in \cref{de:event-structure}.
Verifying that the corresponding families of configurations satisfy the properties of Definition~\ref{de:event-structure} is straightforward.
This section is mainly intended to
provide material for examples and discussions. The reader can quickly
browse through it: only the correspondence with prime event
structures will play a major role in the rest of the paper.

\subparagraph{Prime event structures.}

Prime event structures~\cite{NPW:PNES} are one of the simplest and
most popular event structure models, where dependencies between events
are captured in terms of causality and conflict.

\begin{definition}[prime event structure]
  \label{def:pes}
  A \emph{prime event structure} ({\pes}, for short) is a tuple
  $\ES[P] = \langle E, \leq, \#, \lambda \rangle$, where $E$ is a
  set of events, $\leq$ and $\#$ are binary relations on $E$ called
  \emph{causality} and \emph{conflict}, respectively, and
  $\lambda : E \to \Lambda$ is a labelling function, such that
  \begin{itemize}

  \item $\leq$ is a partial order and $\causes{x} = \{y \in E \mid y
    \leq x\}$ is finite for all $x \in E$;

  \item $\#$ is irreflexive, symmetric and hereditary with respect to
    causality, i.e., for all $x,y,z \in E$, if $x \# y$ and $y \leq z$
    then $x \# z$.
  \end{itemize}
\end{definition}

Configurations are sets of events without conflicts and closed with respect to
causality. For later use, we also introduce a notation for the absence
of conflicts, referred to as consistency.

\begin{definition}[consistency, configuration]
  Let $\ES[P] = \langle E, \leq, \#,\lambda \rangle$ be a {\pes}. We
  say that $x, y \in E$ are \emph{consistent}, written
  $x \cons y$, when $\neg (x \# y)$. A subset $X \subseteq E$ is
  called \emph{consistent}, written $\setcons{X}$, when its elements are
  pairwise consistent. A \emph{configuration} of $\ES[P]$ is a finite
  set of events $C \subseteq E$ such that
  (i) $\setcons{C}$ and
  (ii) for all $x\in C$, $\causes{x}
    \subseteq C$.
\end{definition}

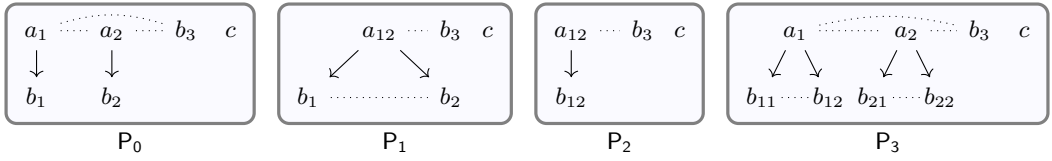
\begin{figure}
  \setlength{\tabcolsep}{4pt}
  \begin{tabular}{cccc}
    \begin{tikzcd}[boxedcd, sep=4mm]
      a_1
      \arrow[d] \arrow[conflict, r]
      \arrow[conflict, bend left=17, rr]
      & a_2 \arrow[d] \arrow[conflict, r] & b_3 &[-3mm] c \\
      b_1           & b_2 
    \end{tikzcd}
    &
    \begin{tikzcd}[boxedcd, row sep=4mm, column sep=3mm]
      & a_{12} \arrow[dl] \arrow[dr]  \arrow[conflict, r] & b_3 &[-3mm] c \\
      b_1 \arrow[conflict, rr]    &     & b_2 
    \end{tikzcd}
    &
    \begin{tikzcd}[boxedcd, row sep=4mm, column sep=3mm]
      a_{12} \arrow[d]  \arrow[conflict, r] & b_3 &[-3mm] c \\
      b_{12}            &  
    \end{tikzcd}
    &
    \begin{tikzcd}[boxedcd, row sep=4mm, column sep=4mm]
      &[-6mm] a_1 \arrow[dr] \arrow[dl] \arrow[conflict, rrr]
      \arrow[conflict, bend left=13, rrrrr]
      &[-6mm] &[-5mm]
      &[-6mm] a_2 \arrow[dr] \arrow[dl]  \arrow[conflict, rr] &[-6mm] &[-5mm]
      b_3 &[-3mm] c\\
      b_{11} \arrow[conflict, shorten <= -3pt, shorten >= -3pt, rr] & & b_{12} &  b_{21} \arrow[conflict, shorten <= -3pt, shorten >= -3pt, rr] & & b_{22} 
    \end{tikzcd}      
    \\
    $\ES[P_0]$ & $\ES[P_1]$ & $\ES[P_2]$ & $\ES[P_3]$ 
  \end{tabular}
  \caption{Some prime event structures}
  \label{fi:pes}
\end{figure}

Some examples of {\pes}s can be found in
Fig.~\ref{fi:pes}. Causality is represented as a solid arrow,
while conflict is represented as a dotted line. For instance, in
$\ES[P_0]$, event $a_1$ is a cause of $b_1$ and it is in conflict both
with $a_2$ and $b_3$. Only direct causalities and
non-inherited conflicts are represented. For instance, in $\ES[P_0]$,
the conflicts $a_1 \# b_2$, $a_2 \# b_1$ and $b_1 \# b_2$ are not represented since they are inherited. The labelling is implicitly represented by naming the events by their label, possibly with some index. For instance, $a_1$ and $a_2$ are events labelled by $a$.

\begin{figure}
\begin{center}
\begin{tikzpicture}[x=16mm,y=17]
  \node at (2,3) (empty) {
    \begin{tikzcd}[poset]
      \ 
    \end{tikzcd}
  };
  
  \node at (1,2) (a12) {
    \begin{tikzcd}[poset]
      a_{12} 
    \end{tikzcd}
  };
  
  \node at (3,2) (c) {
    \begin{tikzcd}[poset]
      c
    \end{tikzcd}
  };

  \node at (5,2) (b3) {
    \begin{tikzcd}[poset]
      b_3
    \end{tikzcd}
  };

  \node at (0,1) (a12b3) {
    \begin{tikzcd}[poset, sep=2mm]
      a_{12} \arrow[d]\\
      b_{12}
    \end{tikzcd}
  };

  \node at (2,1) (a12c) {
    \begin{tikzcd}[poset, sep=2mm]
      a_{12} & c
    \end{tikzcd}
  };

  \node at (4,1) (b3c) {
    \begin{tikzcd}[poset, sep=2mm]
      b_3 & c
    \end{tikzcd}
  };

  \node at (1,0) (a12b3c) {
    \begin{tikzcd}[poset, sep=2mm]
      a_{12} \arrow[d] & c\\
      b_{12}
    \end{tikzcd}
  };
  
  \draw (empty) -- (a12);
  \draw (empty) -- (c);
  \draw (empty) -- (b3);
  \draw (a12)   -- (a12b3);
  \draw (a12)   -- (a12c);
  \draw (c)     -- (b3c);
  \draw (c)     -- (a12c);
  \draw (b3)    -- (b3c);
  \draw (a12b3) -- (a12b3c);
  \draw (a12c)  -- (a12b3c);
\end{tikzpicture}

\end{center}

\caption{The configurations $\conf{\ES[P_2]}$ of the {\pes} $\ES[P_2]$ in Fig.~\ref{fi:pes} viewed as poset event structures}
\label{fi:conf}
\end{figure}
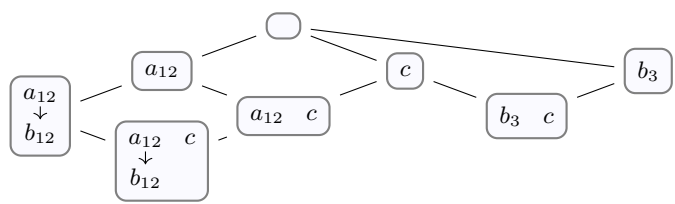

Clearly {\pes}s can be seen as poset event structures. Given a {\pes}
$\ES[P] = \langle E, \leq, \# , \lambda \rangle$ and its set of
configurations $\conf{\ES[P]}$, the local order of a configuration
$C \in \conf{\ES[P]}$ is $\leq_C = \leq \cap (C \times C)$, i.e., the
restriction of the causality relation to $C$. The extension order
turns out to be simply subset inclusion. In fact, given
$C_1 \subseteq C_2$ clearly $\leq_{C_1} = \leq \cap (C_1 \times C_1)$
is the restriction to $C_1$ of
$\leq_{C_2} = \leq \cap (C_2 \times C_2)$.  Moreover, if $x_1 \in C_1$
and $x_2 \in C_2$, with $x_2 \leq_{C_2} x_1$, then necessarily
$x_2 \in C_1$ since configurations are causally closed.
As an example, the {\pes} $\ES[P_2]$ of Fig.~\ref{fi:pes}, viewed as
a poset event structure, can be found in Fig.~\ref{fi:conf}.

\subparagraph{Asymmetric event structures.}

Asymmetric event structures~\cite{BCM:CNAED} are a
generalisation of {\pes} where conflict is allowed to be
non-symmetric.

\begin{definition}[asymmetric event structure]
  \label{de:AES}
  An \emph{asymmetric event structure} ({\aes}, for short) is a tuple
  $\ES[A]= \langle E, \leq, \ac, \lambda \rangle$, where $E$ is a
  set of events, $\leq$ and $\ac$ are binary relations on $E$ called
  \emph{causality} and \emph{asymmetric conflict}, and
  $\lambda : E \to \Lambda$ is a labelling function, such that
  \begin{itemize}
  \item $\leq$ is a partial order and
    $\causes{x} = \{ y \in E \mid y \leq x \}$ is finite for all
    $x \in E$;

  \item $\ac$ satisfies, for all $x, y, z \in E$
    \begin{enumerate}
    \item if $x < y$ then $x \ac y$;
    \item if $x \ac y$ and $y < z$ then $x \ac z$;
    \item $\ac$ is acyclic on $\causes{x}$;
    \item if $\ac$ is cyclic on $\causes{x} \cup \causes{y}$ then $x \ac y$.
    \end{enumerate}
  \end{itemize} 
\end{definition}

In the graphical representation, asymmetric conflict is depicted as
a dotted arrow. For instance, in the asymmetric event structure
$\ES[A_0]$ of Fig.~\ref{fi:aes} we have $a_{12} \ac b_{123}$. Again,
only non inherited asymmetric conflicts are represented.

The asymmetric conflict relation has two natural interpretations,
i.e., $x \ac y$ can be understood as (i) the occurrence of $y$ \emph{prevents} $x$, or (ii) $x$ \emph{precedes} $y$ in all computations
where both appear.
This allows to represent faithfully the existence of precedences between actions and concurrent read accesses to a shared resource (intuitively, while readings can occur concurrently, destructive accesses can follow, but obviously not precede a reading).

The interpretation of asymmetric conflict above should give some
intuition for the conditions in
\cref{de:AES}. Condition (1) naturally arises from
interpretation (ii) above: when $x<y$ clearly $x$ precedes $y$ when
both occur and thus $x \ac y$.
Condition (2) is a form of hereditarity of asymmetric conflict along
causality: if $x \ac y$ and $y < z$ then all runs where $x$
and $z$ appear, necessarily also include $y$, and $x$ precedes $y$
which in turn precedes $z$, hence $x \ac z$.  Concerning
(3) and (4), observe that events forming a cycle of asymmetric
conflict cannot appear in the same run, since each event in the cycle
should occur before itself in the run.
For instance, in the {\aes}
$\ES[A_1]$ of Fig.~\ref{fi:aes}, we have $a_{1} \ac a_2 \ac a_1$, hence $a_1$ and $a_2$ cannot appear in the same computation.
In this view, condition (3) corresponds to irreflexiveness of conflict
in {\pes}s, while condition (4) requires that binary symmetric conflict is
explicitly represented by asymmetric conflict in both directions.
Indeed, prime event structures can be identified with the subclass of {\aes}s where $\ac$ is symmetric.

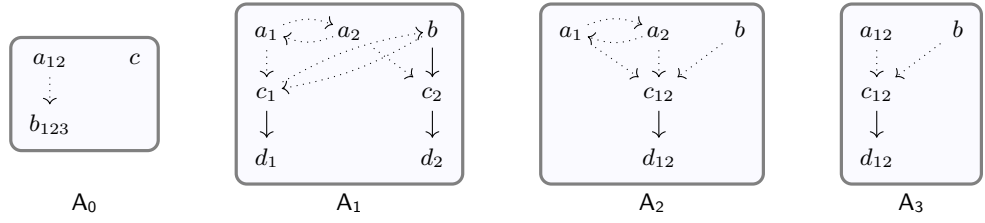
\begin{figure}
  \setlength{\tabcolsep}{14pt}
  \begin{tabular}{cccc}
    \begin{tikzcd}[boxedcd, row sep=4mm, column sep=5mm]
      a_{12} \arrow[ac, d]  & c \\
      b_{123}            &  
    \end{tikzcd}
    &
    \begin{tikzcd}[boxedcd, row sep=4mm, column sep=5mm]
      a_1 \arrow[ac, d]
      \arrow[ac, bend left=20, shorten <= -3pt, shorten >= -3pt, r]     
      & a_{2}
      \arrow[ac, bend left=20, shorten <= -3pt, shorten >= -3pt, l]
      \arrow[dr, ac]
      &  b \ar[d] \arrow[ac, bend left=10, shorten <= -3pt, shorten >= -3pt, dll]
      \\
      c_1 \ar[d]
      \arrow[ac, bend left=10, shorten <= -3pt, shorten >= -3pt, urr]
      & & c_2 \arrow[d]\\
      d_1     &  & d_2
    \end{tikzcd}
    &
    \begin{tikzcd}[boxedcd, row sep=4mm, column sep=5mm]
      a_1 \arrow[ac, dr]
      \arrow[ac, bend left=20, shorten <= -3pt, shorten >= -3pt, r]     
      & a_{2}
      \arrow[ac, bend left=20, shorten <= -3pt, shorten >= -3pt, l]
      \arrow[d, ac]
      &  b \ar[ac, dl]
      \\
      & c_{12} \ar[d]     
      &\\
      & d_{12} &
    \end{tikzcd}
    &
    \begin{tikzcd}[boxedcd, row sep=4mm, column sep=5mm]
      a_{12}
      \arrow[d, ac]
      &  b \ar[ac, dl]
      \\
      c_{12} \ar[d]     
      &\\
      d_{12} &
    \end{tikzcd}
    \\
    $\ES[A_0]$ & $\ES[A_1]$ & $\ES[A_2]$ & $\ES[A_3]$
  \end{tabular}

  \caption{Some asymmetric event structures}
  \label{fi:aes}
\end{figure}

Configurations are again defined as causally closed and
conflict free sets of events.

\begin{definition}[{\aes} configuration]
  Let $\ES[A] = \langle E, \leq, \ac,\lambda \rangle$ be an
  {\aes}. A configuration of $\ES[A]$ is a finite set of events
  $C \subseteq E$ such that\
  (i)~for any $x \in C$, $\causes{x} \subseteq C$ (causally closed)\ \
  (ii)~$\ac$ is acyclic on $C$ (conflict free).
\end{definition}

Also {\aes}s can be  seen as special poset event structures.
Given an {\aes}
$\ES[A] = \langle E, \leq, \ac, \lambda \rangle$ and its set of
configurations $\conf{\ES[A]}$, the local order of a configuration
$C \in \conf{\ES[A]}$ is $\leq_C = (\ac \cap (C \times C))^*$, i.e., the
transitive closure of restriction of the asymmetric conflict to $C$.
The prefix order on configurations is not simply set-inclusion:
since a configuration $C$ cannot be extended with an event which is
prevented by some of the events already present in $C$. Hence for
$C_1, C_2 \in \conf{\ES[A]}$ we have $C_1 \sqsubseteq C_2$ iff
$C_1 \subseteq C_2$ and for all
$x \in C_1, \ y \in C_2 \setminus C_1$, $\lnot(y \ac x)$.
For instance, the configurations $\conf{\ES[A_0]}$ of $\ES[A_0]$,
ordered by prefix, can be obtained from Fig.~\ref{fi:conf}, by
replacing all occurrences of $b_{12}$ and $b_3$, by $b_{123}$. Note,
e.g., that $\{b_{123}\} \not\sqsubseteq \{a_{12}, b_{123}\}$ since
$a_{12} \ac b_{123}$.
  
\subparagraph{Flow event structures.}

In some situations, it can be quite useful to have the possibility of
modelling in a direct way the presence of multiple disjunctive and
mutually exclusive causes for an event, something that is not possible
in {\pes}s and in {\aes}s, where for each event there is a uniquely
determined minimal set of causes. For instance, in a process
calculus with non deterministic choice ``+'' and sequential
composition ``;'' in order to give a {\pes} semantics to $(a+b); c$ we are
forced to use two different events to represent the execution of $c$,
one for the execution of $c$ after $a$ and the other for the execution
of $c$ after $b$.

We briefly describe a model that overcomes this limitation, namely
flow~\cite{BC:SCCS,Bou:FESFN} event structures.

\begin{definition}[flow event structure]
  A \emph{flow event structure} ({\fes}) is a tuple
  $\langle E, \prec, \#, \lambda \rangle$, where $E$ is a
  set of events, $\prec \subseteq E \times E$ is an irreflexive
  relation called the \emph{flow relation},
  $\# \subseteq E \times E$ is the \emph{symmetric conflict} relation, and
  $\lambda : E \to \Lambda$ is a labelling function.
\end{definition}
Causality is replaced by an irreflexive (in general non
transitive) flow relation $\prec$, intuitively representing immediate
causal dependency. Moreover, conflict is no longer hereditary.

An event can have causes which are in conflict and these have a
disjunctive interpretation, i.e., the event will be enabled by a
maximal conflict-free subset of its causes.
This is formalised by the notion of configuration.

\begin{definition}[{\fes} configuration]
  Let $\ES[F] = \langle E, \prec, \#,\lambda \rangle$ be an {\fes}. A
  configuration of $\ES[F]$ is a finite set of events $C \subseteq E$
  such that (i) $\prec$ is acyclic on $C$, (ii) $\neg (x \# x')$ for
  all $x, x' \in C$ and (iii) for all $x \in C$ and $y \notin C$ with
  $y \prec x$, there exists $z \in C$ such that $y \# z$ and
  $z \prec x$.
\end{definition}

Some examples of {\fes}s can be found in Fig.~\ref{fi:fes}. Relation $\prec$
is represented by a double headed solid arrow. For instance, consider
the {\fes} $\ES[F]_1$. The set $C = \{a, d_{01}\}$ is a
configuration. We have $b \prec d_{01}$ and $b \not \in C$, but this
is fine since there is $a \in C$ such that $a \# b$ and
$a \prec d_{01}$.

\begin{figure}
  \begin{center}
  \begin{tabular}{cccc}
    \begin{tikzcd}[boxedcd, row sep=4mm, column sep=7mm]
      a \arrow[flow, d] \arrow[conflict, r] & b \arrow[flow, d] \arrow[conflict, r] & c \arrow[flow, d]\\
      d_0 \arrow[conflict, r] \arrow[conflict, bend right=16, rr]    &  d_1 \arrow[conflict, r]   & d_2 
    \end{tikzcd}
    &
    \begin{tikzcd}[boxedcd, row sep=4mm, column sep=1mm]
      a \arrow[flow, dr] \arrow[conflict, rr] & & b \arrow[flow, dl] \arrow[conflict, rr] & & c \arrow[flow, d]\\
      & d_{01} \arrow[conflict, rrr] & &   &  d_2
    \end{tikzcd}
    &
    \begin{tikzcd}[boxedcd, row sep=4mm, column sep=2mm]
      a \arrow[flow, d] \arrow[conflict, rr] & & b \arrow[flow, dr] \arrow[conflict, rr] & & c \arrow[flow, dl]\\
      d_0 \arrow[conflict, rrr] & & & d_{12}  & 
    \end{tikzcd}
    &
    \begin{tikzcd}[boxedcd, row sep=4mm, column sep=2mm]
      a \arrow[flow, dr] \arrow[conflict, r] & b \arrow[flow, d] \arrow[conflict, r] & c \arrow[flow, dl]\\
      & d_{012} &
    \end{tikzcd}                                                
     \\
    $\ES[F_0]$ & $\ES[F_1]$ & $\ES[F_2]$ & $\ES[F_3]$
  \end{tabular}
  \end{center}
  \caption{Some flow structures}
  \label{fi:fes}
\end{figure}
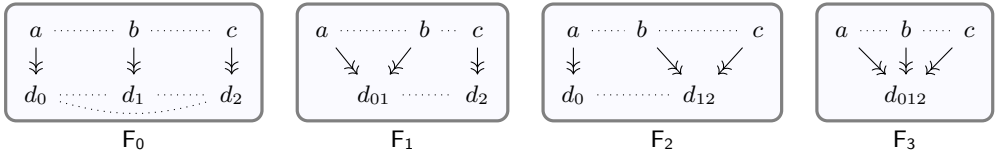

Under mild assumptions that exclude the presence of non-executable
events (a condition referred to as fullness
in~\cite{Bou:FESFN}), {\fes}s can be seen as poset event
structures,
by endowing each configuration $C$ with a local order arising as the reflexive and transitive closure of the
restriction of the flow relation to $C$, i.e., $\leq_C = (\prec \cap (C \times C))^*$.

\subparagraph{Bundle event structures.}
  
Bundle event structures~\cite{Lan92:BES,Lan:TASF} are another event
structure model that has been introduced in order to enable a direct
representation of disjunctive causes, thus easing the
definition of the semantics of the process description language
\textsc{lotos}.

\begin{definition}[bundle event structure]
  \label{de:bundle}
  A \emph{bundle event structure} is a triple $\langle E, \mapsto, \#
  \rangle$, where $E$ is the denumerable set of events, $\# \subseteq
  E \times E$ is the (irreflexive) \emph{symmetric conflict} relation
  and $\mapsto \subseteq \Powfin{E} \times E$ is the \emph{bundle
    relation} such that $X \times X  \subseteq \#$.
\end{definition}

Here a set of multiple disjunctive and mutually exclusive causes for
an event is called a \emph{bundle set} for the event, and comes into
play as a primitive notion.  The explicit representation of the
bundles makes bundle event structures strictly less expressive than
flow event structures.
(see~\cite{Lan:TASF}
for a wider discussion).
On the other hand, bundle event structures offer the advantage of
having a simpler theory. For instance, differently from what happens
for flow event structures, non-executable events can be removed
without affecting the behaviour of the event structure.

Configurations can be defined as conflict free sets of events that
contain, for each event, a element from of each of its bundles. Formally, $C \subseteq E$ is a configuration if (i) the relation $\mapsto_C$ defined, for $x, y \in C$, by $x \mapsto_C y$ when $X \mapsto y$ and $x \in X$, is acyclic (ii) $\neg (x \# y)$ for all $x, y \in C$; (iii) $X \cap C \neq \emptyset$ for all $x \in C$ and $X \mapsto x$. Endowing configurations with  $\mapsto_C^*$ turns a bundle event structure into an event structure in the sense of \cref{de:event-structure}.

\section{Foldings of Event Structures}
\label{se:folding}

In this section, we study a notion of folding, which is intended to
formalise the intuition of a behaviour-preserving quotient for an
event structure. We prove that there  always exists a minimal quotient
and we show that foldings between general poset event structures always arise,
in a suitable formal sense, from foldings over prime event structures.

\subsection{Morphisms and Foldings}

We first endow event structures with a notion of morphism. Below,
given two event structures $\ES$, $\ES'$, a function $f : E \to E'$
and a configuration $C \in \conf{\ES}$, we write $f(C)$ to refer to
the configuration whose underlying set is $\{ f(x) \mid x \in C \}$,
endowed with the order $f(x) \leq_{f(C)} f(y)$ iff $x \leq y$.

\begin{definition}[morphism]
  \label{de:es-morphism}
  Let $\ES, \ES'$ be event structures. A \emph{(strong) morphism}
  $f : \ES \to \ES'$ is a function between the underlying sets of
  events such that $\lambda = \lambda' \circ f$ and for all
  configurations $C \in \conf{\ES}$, the function $f$ is injective on
  $C$ and $f(C) \in \conf{\ES'}$.
\end{definition}

Hereafter, the qualification ``strong'' will be omitted since this is the
only kind of morphisms we deal with. It is motivated by the fact that
normally morphisms on event structures are designed to represent
simulations. If this were the purpose, then the requirement on
preservation of configurations could have been weaker, i.e., we
could have asked the order in the target configuration to be included
in (not identical to) the image of the order of the source
configuration (precisely, given a configuration
$\langle C, \leq_C \rangle \in \conf{\ES}$ then there exists
$\langle C', \leq_{C'} \rangle \in \conf{\ES'}$ such that $C' = f(C)$
and for all $x, y \in C$, $f(x) \leq_{C'} f(y)$ implies $x \leq_C y$).
Moreover, morphisms could have been partial.
However, in our setting, for the
objective of defining history-preserving quotients, the stronger
notion works fine and simplifies the presentation.

\begin{remark}
  \emph{The composition of morphisms is a morphism and the identity is
    a morphism. Hence the class of event structures and event structure
    morphisms form a category $\cat{ES}$.}
\end{remark}

\begin{definition}[folding]
  Let $\ES$ and $\ES'$ be event structures. A
  \emph{folding} is a morphism $f : \ES \to \ES'$
  such that the relation
  $R_f = \{ (C, f_{|C}, f(C)) \mid C \in \conf{\ES}\}$
  is a hhp-bisimulation.
\end{definition}
In words, a folding is a function that ``merges'' some sets of events
of an event structure into single event without altering the
behaviour modulo hhp-bisimilarity.
In~\cite{ABG:RESHPB} the notion of folding asks for the preservation of hp-bisimilarity, a weaker behavioural equivalence which is defined as hhp-bisimilarity but omitting the requirement of downward-closure.
Note that, as far as the notion of folding is concerned, this makes no difference: $R_f$ is downward-closed by definition, hence it is a hhp-bisimulation whenever it is a hp-bisimulation. Instead, taking hhp-bisimilarity as the reference equivalence appears to be the right choice for the development of the theory. E.g., it allows one to prove Lemma~\ref{le:bisim-as-es} that plays an important role for arguing about the adequateness of the notion of folding.
Interestingly, foldings can be characterised as open maps in the sense
of~\cite{JNW:BFOM}, by taking conflict free prime event structures as
subcategory of observations. This is explicitly worked out in the
appendix (\cref{le:folding-open}).

As an example, consider the {\pes}s in Fig.~\ref{fi:pes} and 
the function $f_{02} : \ES[P_0] \to \ES[P_2]$ that maps events as suggested by the indices, i.e., $f_{02}(a_1) = f_{02}(a_2) = a_{12}$,
$f_{02}(b_1) = f_{02}(b_2) = b_{12}$, $f_{02}(b_3) = b_3$ and $f_{02}(c)=c$.
Then it is easy to see that $f_{02}$ is a folding.
Note that, instead, $f_{01} : \ES[P_0] \to \ES[P_1]$, again mapping events according to their indices, is not a folding. In fact, $f_{01}(\{a_1\}) = \{a_{12}\} \trans{b_2} \{ a_{12}, b_2\}$, but clearly there is no transition $\{a_1\} \trans{x}$ with $f_{01}(x) = b_2$, since the only counter-image of $b_2$ in $\ES[P_0]$ is $b_2$.

It is also interesting to observe that the greater expressiveness of
{\aes}s allows one to obtain smaller quotients. For instance, while
the {\pes} $\ES[P_2]$ in Fig.~\ref{fi:pes} is minimal in the class of
{\pes}s, if we view it as an {\aes}, it can be further reduced. In
fact the obvious function from $\ES[P_2]$ to the {\aes} $\ES[A_0]$ in
Fig.~\ref{fi:aes} can be easily seen to be a folding.

\begin{remark}
  \label{le:fold-comp}
  \emph{The composition of foldings is a folding and the identity is a
    folding. We can consider a subcategory $\cat{ES_f}$ of $\cat{ES}$
    with the same objects and foldings as
    morphisms (see~\cref{le:folding-composition} in the Appendix).}
\end{remark}

Again in the setting of {\aes}s, consider the structures in
Fig.~\ref{fi:aes} and the functions $g_{12} : \ES[A_1] \to \ES[A_2]$,
and $g_{23} : \ES[A_2] \to \ES[A_3]$, naturally induced by the
indices. These can be seen to be foldings. The first one merges $c_1$, in
conflict with $b$ and $c_2$ caused by $b$ to a single event $c_{12}$,
in asymmetric conflict with $b$. The second one merges the two
conflicting events $a_1$ and $a_2$ into a single one $a_{12}$. Their
composition $g_{13} = g_{23} \circ g_{12} : \ES[A_1] \to \ES[A_3]$ is
again a folding.

Consider the {\fes}s in Fig.~\ref{fi:fes}. Again the obvious functions
from $\ES[F_0]$ to $\ES[F_1]$ and $\ES[F_2]$ can be seen to be
foldings. Instead, seen as a {\pes}, the event structure $\ES[F_0]$ is
minimal.

The next result shows that if we know that $f : \ES[E] \to \ES[E']$ is a
morphism, then half of the conditions needed to be a hhp-bisimulation
and thus folding, i.e., condition (1) in \cref{de:hp-bisim},
is automatically satisfied. This is used later in proofs whenever we need to show that some map is a folding.

\begin{restatable}[from morphisms to foldings]{lemma}{lemorphtofold}
  \label{le:morph-to-fold}
  Let $\ES$ and $\ES'$ be event structures and let $f : \ES \to \ES'$
  be a morphism.
  If for all $C_1 \in \conf{\ES}$ and transition
  $f(C_1) \trans{x'} C_2'$ there exists $C_1 \trans{x} C_2$ such that
  $f(C_2) = C_2'$ then $f$ is a folding.
\end{restatable}

A simple but crucial result shows that the target event structure for
a folding is completely determined by the mapping on events. This
allows us to view foldings as equivalences on the source event
structures. We first define the quotient induced by a morphism.

\begin{definition}[quotients from morphisms]
  \label{de:morph-quotient}
  Let $\ES$, $\ES'$ be event structures and let $f : \ES \to \ES'$ be
  a morphism. Let $\equiv_f$ be the equivalence relation on $\ES$ defined by
  $x \equiv_f y$ if $f(x)=f(y)$. We denote by $\quotient{\ES}{\equiv_f}$  the
  event structure with configurations
  $\conf{\quotient{\ES}{\equiv_f}} = \{ \eqclass{C}{\equiv_f} \mid C
  \in \conf{\ES} \}$ where
  $\eqclass{C}{\equiv_f} = \{ \eqclass{x}{\equiv_f} \mid x \in C \}$
  is ordered by
  $\eqclass{x}{\equiv_f} \leq_{\eqclass{C}{\equiv_f}}
  \eqclass{y}{\equiv_f}$ iff $x \leq_{C} y$.
\end{definition}

It is immediate to see that $\quotient{\ES}{\equiv_f}$ is a
well-defined event structure.

\begin{restatable}[folding as equivalences]{lemma}{lefoldisequivalence}
  \label{le:fold-is-equivalence}
  Let $\ES$, $\ES'$ be event structures and let $f : \ES \to \ES'$ be
  a morphism. If $f$ is a folding then $\quotient{\ES}{\equiv_f}$ is
  isomorphic to $\ES'$.
\end{restatable}

The previous result allows us to identify foldings with the
corresponding equivalences on the source event structures and
motivates the following definition.

\begin{definition}[folding equivalences]
  \label{de:foldeq}
  Let $\ES$ be an event structure. The set of folding equivalences
  over $\ES$ is
  $\foldeq{\ES} = \{ \equiv_f \mid f : \ES \to \ES'\ \mbox{folding for some $\ES'$}\}$.
\end{definition}

Hereafter, we will freely switch between the two views of foldings as
morphisms or as equivalences, since each will be convenient for some
purposes.

We next observe that given two foldings we can always take their
``join'', providing a new folding that, roughly speaking, produces a smaller
quotient than both the original ones.

\begin{restatable}[joining foldings]{proposition}{prfoldjoin}
  \label{pr:fold-join} 
  Let $\ES, \ES', \ES''$ be event structures and let
  $f' : \ES \to \ES'$, $f'' : \ES \to \ES''$ be foldings. Define
  $\ES'''$ as the quotient $\quotient{\ES}{\equiv}$ where
  $\equiv$ is the transitive closure of
  $\equiv_{f'} \cup \equiv_{f''}$.
  Then $g': \ES' \to \ES'''$ defined by $g'(x') = \eqclass{x}{\equiv}$
  if $f'(x)=x'$ and $g'': \ES'' \to \ES'''$ defined by
  $g''(x'') = \eqclass{x}{\equiv}$ if $f''(x)=x''$ are foldings.
\end{restatable}

As an example, consider the {\pes} in Fig.~\ref{fi:pes} and two
morphisms $f_{30} : \ES[P_3] \to \ES[P_0]$ and
$f_{31} : \ES[P_3] \to \ES[P_1]$. The way all events are mapped by
$f_{30}$ and $f_{31}$ is naturally suggested by their labelling, apart
for the $b_{ij}$ for which we let $f_{30}(b_{ij}) = b_i$ while
$f_{31}(b_{ij}) = b_j$.
It can be seen that both are
foldings. Their join, constructed as in \cref{pr:fold-join}, is
$\ES[P_2]$ with the folding morphisms $f_{02} : \ES[P_0] \to \ES[P_2]$
and $f_{12} : \ES[P_1] \to \ES[P_2]$.

\begin{remark}
  \emph{\cref{pr:fold-join} is a consequence of the fact that the
    category $\cat{ES}$ has pushouts of foldings. Indeed, $\ES'''$ as
    defined above is the pushout of $f'$ and $f''$ (in $\cat{ES}$ and
    also in $\cat{ES_f}$). It can be seen that, instead, $\cat{ES}$
    does not have all pushouts (see Fig.~\ref{fi:no-pushout} in the
    Appendix for a counterexample).}
\end{remark}

When interpreted in the set of folding equivalences of an event
structure, \cref{pr:fold-join} has a clear meaning. Recall that
the equivalences over some fixed set $X$, ordered by inclusion, form a
complete lattice, where the top element is the universal equivalence
$X \times X$ and the bottom is the identity on $X$. Then
\cref{pr:fold-join} implies that $\foldeq{\ES}$ is a
sublattice of the lattice of equivalences.
Actually, it can be shown that $\foldeq{\ES}$ is itself a complete
lattice. Therefore each event structure $\ES$ admits a maximally
folded version.

\begin{restatable}[lattice of foldings]{proposition}{prfoldinglattice}
  \label{pr:folding-lattice}
  Let $\ES$ be an event structure. Then $\foldeq{\ES}$ is a sublattice
  of the complete lattice of equivalence relations over $\ES$. 
\end{restatable}

\begin{remark}
  \emph{The above result arises from a generalisation of
    \cref{pr:fold-join} showing that, for any event structure $\ES$,
    each collection of foldings $f_i : \ES \to \ES_i$, with $i \in I$, admits a colimit
    in $\cat{ES}$. Thus the coslice category
    $(\ES \downarrow \cat{ES_f})$ has a terminal object, which is the
    maximally folded event structure.}
\end{remark}

It is natural to ask whether all behaviour preserving
quotients correspond to foldings. Strictly speaking, the answer is negative. More
precisely, there can be morphisms $f : \ES \to \ES'$ such that
$\quotient{\ES}{\equiv_f}$ is hhp-bisimilar to $\ES$, but $f$ is not a
folding. For an example, consider the {\pes}s $\ES[P_0]$ and
$\ES[P_1]$ in Fig.~\ref{fi:pes} and the morphism
$f_{01} : \ES[P_0] \to \ES[P_1]$ suggested by the indexing. We already
observed this is not a folding, but
$\quotient{\ES[P_0]}{\equiv_{f_{01}}}$, which is isomorphic to $\ES[P_1]$, is
hhp-bisimilar to $\ES[P_0]$.

However, we can show that for any behaviour preserving quotient, there
is a folding that produces a coarser equivalence, and thus a
smaller quotient. For instance, in the example discussed above, there is the
folding $f_{02} : \ES[P_0] \to \ES[P_2]$, that ``produces'' a smaller
quotient.

This follows from the possibility of joining foldings
(\cref{pr:fold-join}) and the fact that a hhp-bisimulation can be
always seen as an event structure, a result that generalises to our
setting a property proved for {\pes}s in~\cite{Ber:HHP}.%

\begin{restatable}[foldings subsume behavioural quotients]{proposition}{prfoldsubsume}
  \label{pr:fold-subsume}
  Let $\ES$ be an event structure and let $f : \ES \to \ES'$ be a
  morphism such that $\quotient{\ES}{\equiv_f}$ is hhp-bisimilar to
  $\ES$. Then there exists a folding $g : \ES \to \ES''$ such that
  $\equiv_g$ is coarser than $\equiv_f$.
\end{restatable}

\subsection{Folding through Prime Event Structures}

Here we observe that each event structure is the folding of some
canonical {\pes}. We then prove that, interestingly enough, all
foldings between event structures arise from foldings of the
corresponding canonical {\pes}s.

We start with the definition of the canonical {\pes} associated with an
event structure.

\begin{definition}[{\pes} for an event structure]
  \label{de:pes-es}
  Let $\ES$ be an event structure. Its canonical {\pes} is
  $\pr(\ES) = \langle \Hist{\ES}, \sqsubseteq, \#, \lambda' \rangle$
  where $\sqsubseteq$ is prefix, $\#$ is inconsistency, i.e., for
  $H_1, H_2 \in \Hist{\ES}$ we let $H_1 \# H_2$ if
  $\neg (H_1 \cons H_2)$ and $\lambda'(H) = \lambda(x)$ when
  $H \in \Hist{x}$. Given a morphism $f : \ES \to \ES'$ we write
  $\pr(f) : \pr(\ES) \to \pr(\ES')$ for the morphism defined by
  $\pr(f)(H) = f(H)$.
\end{definition}

It can be easily seen that the definition above is well-given. In
particular, $\pr(\ES)$ is a well-defined {\pes} because, as proved in~\cite{Ren:PC}, a family of posets ordered by
prefix is finitary coherent prime algebraic domain. Then the tight
relation between this class of domains and {\pes} highlighted
in~\cite{Win:ES} allows one to conclude the proof.
For instance, in Fig.~\ref{fi:es}(right) one can find the canonical
{\pes} for the event structure on the left.

The canonical {\pes} associated with an event structure can always be folded to
the original event structure.

\begin{restatable}[unfolding event structures to {\pes}'s]{lemma}{leunfespes}
 \label{le:unf-es-pes}
  Let $\ES$ be an event structure. Define a function
  $\phi_{\ES} : \pr(\ES) \to \ES$, for all $H \in \Hist{\ES}$ by
  $\phi_{\ES}(H) = x$ if $H \in \Hist{x}$ for $x \in \ES$.  Then
  $\phi_{\ES}$ is a folding.
\end{restatable}

We next show that any morphism and any folding from a {\pes} to an
event structure $\ES$ factorises uniquely through the {\pes}
$\pr(\ES)$ associated with $\ES$ (categorically, $\phi_{\ES}$ is cofree over $\ES$). This will be useful to relate
foldings in $\ES$ with foldings in $\pr(\ES)$.

\begin{restatable}[cofreeness of $\phi_{\ES}$]{lemma}{lefactorfold}
  \label{le:factor-fold}
  Let $\ES$ be an event structure, let $\ES[P']$ be a {\pes} and let
  $f : \ES[P'] \to \ES$ be an event structure morphism. Then there exists
  a unique morphism $g : \ES[P'] \to \pr(\ES)$ such that
  $f = \phi_{\ES} \circ g$.
  \begin{center}
    \begin{tikzcd}[ampersand  replacement=\&, row sep=5mm]
      \pr(\ES) \arrow[r, "\phi_{\ES}" above]
      \&
      \ES\\
      \ES[P'] \arrow[ur, "f" below] \arrow[u, dotted, "g" left]
      \& {}
    \end{tikzcd}
  \end{center}
  Moreover, when $f$ is a folding then so is $g$.
\end{restatable}

\begin{remark}
  \emph{\cref{le:factor-fold} means that the category $\cat{PES}$ of
    prime event structures is a coreflective subcategory of $\cat{ES}$,
    i.e., $\pr : \cat{ES} \to \cat{PES}$ can be seen as a functor,
    right adjoint to the inclusion
    $\mathbb{I} : \cat{PES} \to \cat{ES}$. Moreover, $\pr$ restricts
    to a functor on the subcategory of foldings,
    $\pr : \cat{ES_f} \to \cat{PES_f}$, where an analogous result
    holds.}
\end{remark}

We conclude that all foldings between event structures arise from
foldings of the associated {\pes}s. Given that $\cat{PES}$ is a
coreflective subcategory of $\cat{ES}$ and foldings can be seen as open
maps, this result (and also the fact that morphisms $\phi_{\ES}$ are
foldings) can be derived from~\cite[Lemma 6]{JNW:BFOM}. The appendix
gives more details on this point (and also reports a direct proof).

\begin{restatable}[folding through {\pes}s]{proposition}{prfoldespes}
  \label{pr:fold-es-pes}
  Let $\ES, \ES'$ be event structures. For all morphisms
  $f : \ES \to \ES'$ consider $\pr(f) : \pr(\ES) \to \pr(\ES')$
  defined by $\pr(f)(H) = f(H)$. Then $f$ is a folding iff $\pr(f)$ is a
  folding.
\end{restatable}

\section{Foldings for Prime and Asymmetric  Event Structures}
\label{se:folding-criteria}

In this section we study foldings on specific subclasses of poset
event structures, providing suitable characterisations. Motivated by
the fact that foldings on general poset event structures always arise
from foldings of the corresponding canonical {\pes}s we first and
mainly focus on {\pes}s. Then we discuss how this can be extended to
asymmetric event structures (and only give a hint to
flow and  bundle event structures). We will see that while {\pes}s admit a least
folding, the other classes of event structures do not.

\subsection{Folding Prime Event Structures}

Since foldings are special morphisms, we first provide a
characterisation of {\pes} morphisms.

\begin{restatable}[{\pes} morphisms]{lemma}{lepesmorphism}
  \label{le:pes-morphism}
  Let $\ES[P]$ and $\ES[P']$ be {\pes}s and let $f : P \to P'$ be a
  function on the underlying sets of events. Then $f$ is a morphism
  iff for all $x, y \in \ES[P]$
  \begin{enumerate}
  \item \label{le:pes-morphism:1}
    $\lambda'(f(x)) = \lambda(x)$;
    
  \item \label{le:pes-morphism:2}
    $f(\causes{x}) = \causes{f(x)}$; namely (a) for all $x' \in \ES[P']$, if $x' \leq f(y)$
      there exists $x \in \ES[P]$ such that $x \leq y$ and $f(x) = x'$
      (b) if $x \leq y$ then $f(x) \leq f(y)$;
      
    \item \label{le:pes-morphism:3}
      (a) if $f(x) = f(y)$ and $x \neq y$ then $x \# y$ and
      (b) if $f(x) \# f(y)$ then $x \# y$.    
  \end{enumerate}
\end{restatable}

These are the standard conditions characterising (total) {\pes}
morphisms (see, e.g.,~\cite{Win:ES}), with the addition of condition
(2b) that is imposed to ensure that configurations are mapped to
isomorphic configurations, as required by the notion of (strong)
morphism (\cref{de:es-morphism}).

We know that not all {\pes} morphisms are foldings.
We next identify some additional conditions characterising those
morphisms which are foldings.

\begin{restatable}[{\pes} foldings]{proposition}{prpesfolding}
  \label{pr:pes-folding}
  Let $\ES[P]$ and $\ES[P']$ be {\pes}s and let
  $f : \ES[P] \to \ES[P']$ be a morphism. Then $f$ is a folding if and
  only if it is surjective  and  for all $X, Y \subseteq \ES[P]$, $x, y \in \ES[P]$, $y' \in \ES[P']$
  \begin{enumerate}
  \item \label{pr:pes-folding:1}
    if $x \#^\forall f^{-1}(y')$ then $f(x) \# y'$;
    
  \item \label{pr:pes-folding:2}
    if $\setcons{(X \cup \{ x\})}$, $\setcons{(Y \cup \{ y\})}$,
    $\setcons{(X \cup Y)}$ and $f(x) = f(y)$ then there exists
    $z \in \ES[P]$ such that $f(z)=f(x)$ and
    $\setcons{(X \cup Y \cup \{ z\})}$.
  \end{enumerate}
\end{restatable}

The notion of folding on {\pes}s turns out to be closely related
to that of abstraction homomorphism for {\pes}s
introduced in~\cite{Cas:phd} for similar purposes. More precisely, abstraction homomorphisms can be
characterised as those {\pes} morphisms additionally satisying
condition (\ref{pr:pes-folding:1}) of \cref{pr:pes-folding}, while they do not necessarily satisfy condition (\ref{pr:pes-folding:2}). Their more liberal definition is explained by the fact that they are designed to work on a subclass of structured {\pes}s (see \cref{le:abs-fold} in the Appendix).

We finally show what the conditions characterising foldings
look like when transferred to equivalences.

\begin{restatable}[folding equivalences for {\pes}s]{corollary}{cofoldeqpes}
  \label{co:fold-eq-pes}
  Let $\ES[P]$ be a {\pes} and let $\equiv$ be an equivalence on
  $\ES[P]$. Then $\equiv$ is a folding equivalence in
  $\foldeq{\ES[P]}$ iff for all $x, y \in \ES[P]$, if $x \equiv y$
  then
  \begin{enumerate}
  \item $\lambda(x) = \lambda(y)$;\ \
  \item $\eqclass{\causes{x}}{\equiv} = \eqclass{\causes{y}}{\equiv}$;\ \
  \item $x \# y$.
  \end{enumerate}
  Moreover, for all
  $x, y \in \ES[P]$, $X, Y \subseteq \ES[P]$
  \begin{enumerate}
    \setcounter{enumi}{3}
  \item if $x \#^\forall \eqclass{y}{\equiv}$ then $\eqclass{x}{\equiv} \#^\forall \eqclass{y}{\equiv}$;
  \item if $\setcons{(X \cup \{ x\})}$, $\setcons{(Y \cup \{ y\})}$,
    $\setcons{(X \cup Y)}$ there exists
    $z \in \eqclass{x}{\equiv}$ such that
    $\setcons{(X \cup Y \cup \{ z\})}$.
  \end{enumerate}
\end{restatable}

For instance, in Fig.~\ref{fi:pes}, consider the equivalence $\equiv_{01}$ over $\ES[P_0]$ such that $a_1 \equiv_{01} a_2$. This produces $\ES[P_1]$ as quotient. This is not a folding equivalence since condition (4) fails: $a_1 \#^\forall \eqclass{b_2}{\equiv_{01}}$, but $\neg (a_2 \# b_2)$ and thus $\neg(\eqclass{a_1}{\equiv_{01}} \#^\forall \eqclass{b_2}{\equiv_{01}})$. Instead, the equivalence $\equiv_{02}$ over $\ES[P_0]$ such that $a_1 \equiv_{02} a_2$ and $b_1 \equiv_{02} b_2$, producing $\ES[P_2]$ as quotient, satisfies all five conditions. 

When {\pes}s are finite, the result above suggests a possible
way of identifying foldings: one can pair candidate events to be
folded on the basis of conditions (1)-(3) and then try to extend the
sets with condition (4)-(5) when possible. The procedure can be
inefficient due to the global flavor of the conditions. This will be further
discussed in the conclusions.

We know from \cref{pr:fold-join} that all event structures admit a ``maximally folded'' version. We next observe that the same result holds in the class of  {\pes}s, i.e., that for each {\pes} there is a uniquely determined minimal
quotient.

\begin{restatable}[joining foldings on {\pes}'s]{lemma}{lefoldjoinpes}
  \label{le:fold-join-pes}
  Let $\ES[P], \ES[P'], \ES[P'']$ be {\pes}s and let
  $f' : \ES[P] \to \ES[P']$, $f'' : \ES[P] \to \ES[P'']$ be
  foldings. Define $\ES'''$ along with $g' : \ES[P'] \to \ES'''$ and
  $g'': \ES[P''] \to \ES'''$ as in \cref{pr:fold-join}. Then
  $\ES'''$ is a {\pes}.
\end{restatable}

\begin{remark}
  \emph{\cref{le:fold-join-pes} is a consequence of the fact that the
    subcategory $\cat{PES_f}$ is a coreflective subcategory of
    $\cat{ES_f}$ and thus it is closed under pushouts.}
\end{remark} 

\subsection{Folding Asymmetric Event Structures}%

We know that foldings on all poset event structures arise from foldings on
the corresponding canonical {\pes}s. Still, for theoretical purposes
and for efficiency reasons, a direct approach, not requiring the
generation of the associated {\pes}, can be of interest. Here we
explicitly discuss the case of asymmetric event structures.
This generalises the results in~\cite{ABG:RESHPB} that identify conditions which are only sufficient and apply to a subclass of foldings (the so-called called elementary foldings, merging a single set of events).
Note also that, despite the fact that in this paper we work in a
slightly different framework, we continue to have that, as observed
in~\cite{ABG:RESHPB}, {\aes}s (and also {\fes}s) do not admit a
unique minimal quotient in general.

We first characterise morphisms in the sense of \cref{de:es-morphism} on {\aes}s.

\begin{restatable}[{\aes} morphisms]{lemma}{leaesmorphism}
  \label{le:aes-morphism}
  Let $\ES[A]$ and $\ES[A']$ be {\aes}s and let
  $f : A \to A'$ be a function on the underlying sets of events. Then $f$ is a morphism if and
  only if for all $x, y \in \ES[A]$, $x \neq y$
  \begin{enumerate}
  \item \label{le:aes-morphism:1}
    $\lambda(f(x)) = \lambda(x)$;

  \item \label{le:aes-morphism:2}
    $\causes{f(x)} \subseteq f(\causes{x})$;
        
  \item \label{le:aes-morphism:3}
    (a) if $f(x) \ac f(y)$ then $x \ac y$ and (b) if $x \ac y$ and $\neg (y \ac x)$ then $f(x) \ac f(y)$;
  \item \label{le:aes-morphism:4}
    if $f(x) = f(y)$ then $x \ac y$.
  \end{enumerate}
\end{restatable}

These are the standard conditions characterising (total) {\aes}
morphisms (see~\cite{BCM:CNAED}), with the addition of
(\ref{le:aes-morphism:3}b), needed in order to ensure that
configurations are mapped to isomorphic configurations.

\begin{restatable}[{\aes} foldings]{proposition}{praesfolding}
  \label{pr:aes-folding}
  Let $\ES[A]$ and $\ES[A']$ be {\aes}s and let
  $f : \ES[A] \to \ES[A']$ be a morphism. Then $f$ is a folding if and
  only if it is surjective and for all $X, Y \subseteq \ES[A]$,
  $x, y \in \ES[A]$ with $x \notin X$, $y \notin Y$, $y' \in \ES[A']$
  \begin{enumerate}
  \item \label{le:aes-folding:1}
    if $f^{-1}(y') \ac^\forall x$ 
    then $y' \ac^\exists f(\causes{x})$;
    
  \item \label{le:aes-folding:2}
    if $\neg (x \ac^\exists X)$,
    $\neg (y \ac^\exists Y )$,
    $\setcons{(X \cup Y)}$ and $f(x) = f(y)$ then there exists
    $z \in \ES[A]$ such that $f(z)=f(x)$ and
    $\neg (z \ac^\exists X \cup Y)$.
    
  \item \label{le:aes-folding:3}
    given $H \in \Hist{x}$, if $\neg (H \ac^\exists X)$, and
    $H_1 \sqsubsetneq H$ such that
    $f(H_1) \cup \{f(x)\} \in \Hist{f(x)}$ there exists $x_1$ such
    that $H_1 \cup \{ x_1 \} \in \Hist{x_1}$ and $\neg (x_1 \ac^\exists X)$.
  \end{enumerate}
\end{restatable}

We already observed that working in the class of {\aes}s we can obtain smaller quotients than in the class of {\pes}s (see, e.g., the hhp-bisimilar structures $\ES[P_2]$ in Fig.~\ref{fi:pes} and $\ES[A_0]$ in Fig.~\ref{fi:aes}). However, not unexpectedly, the folding criteria for {\aes}s are less elegant and more complex than those for {\pes}s. In a practical use, the reference to histories could cause a loss of efficiency.
Moreover, the uniqueness of the minimal quotient is lost. Consider for
instance the {\aes}s in Fig.~\ref{fi:aes-non-min}. It can be seen that
$h_{01} : \ES[A_0] \to \ES[A_1]$ is a folding where the events $c_1$, caused
by $a$ and $c_0$ in conflict with $a$, are merged in a single event
$c_{01}$ in asymmetric conflict with $a$. Similarly, 
$h_{02} : \ES[A_0] \to \ES[A_2]$ is a folding obtained by merging $c_0$ and $c_2$. These are two
minimal foldings that do not admit a join in the class of
{\aes}s. In fact, if we merge all three $c$-labelled events we obtain
$\ES[A_3]$, and it is easy to see that the function
$h_{03} : \ES[A_0] \to \ES[A_3]$ is not a folding. In fact, consider
$\{ a, b \} \in \conf{\ES[A_0]}$. Then
$h_{03}(\{ a, b\}) = \{a,b\} \trans{c_{012}}$, a transition that
cannot be simulated in $\ES[A_0]$. Indeed, it can be seen that the join of $h_{01}$ and $h_{02}$ is the event structure $\ES$ in Fig.~\ref{fi:es}(right), which cannot be represented as an {\aes}.

\begin{figure}
  \setlength{\tabcolsep}{7pt}
  \begin{tabular}{cccc}    
    \begin{tikzcd}[boxedcd, row sep=6mm, column sep=7mm]
      a \arrow[d]
      \arrow[ac, bend left=10, r]
      \arrow[ac, bend left=10, rrd]
      & c_0
      \arrow[ac, bend left=10, l]
      \arrow[ac, bend left=10, r]     
      &  b \ar[d]
      \arrow[ac, bend left=10, l]
      \arrow[ac, bend left=10, lld]     
      \\
      c_1
      \arrow[ac, bend left=10, urr]
      & &
      c_2
      \arrow[ac, bend left=10, ull]
    \end{tikzcd}
    &
    \begin{tikzcd}[boxedcd, row sep=6mm, column sep=7mm]
      a \arrow[ac, d]
      \arrow[ac, bend left=10, rrd]
      &     
      &  b \ar[d]
      \arrow[ac, bend left=10, lld]     
      \\
      c_{01}
      \arrow[ac, bend left=10, urr]
      & &
      c_2
      \arrow[ac, bend left=10, ull]
    \end{tikzcd}
    &
    \begin{tikzcd}[boxedcd, row sep=6mm, column sep=7mm]
      a \arrow[d]
      \arrow[ac, bend left=10, rrd]
      &     
      &  b \ar[ac, d]
      \arrow[ac, bend left=10, lld]     
      \\
      c_1
      \arrow[ac, bend left=10, urr]
      & &
      c_{02}
      \arrow[ac, bend left=10, ull]
    \end{tikzcd}        
    &
    \begin{tikzcd}[boxedcd, row sep=6mm, column sep=3mm]
      a \arrow[ac,dr]
      &     
      &  b \ar[ac, dl]
      \\
      & c_{012} &
    \end{tikzcd}        
    \\
    $\ES[A_0]$ & $\ES[A_1]$ & $\ES[A_2]$ & $\ES[A_3]$
  \end{tabular}

  \caption{Asymmetric event structures do not admit a minimal quotient}
  \label{fi:aes-non-min}
\end{figure}
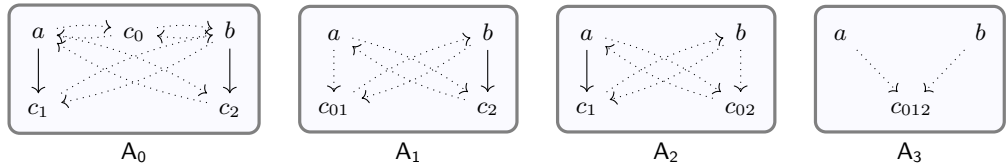

In passing, we note that also in the class of {\fes}s the existence of
minimal foldings is lost. In fact, consider Fig.~\ref{fi:fes}. It can
be easily seen that $\ES[F_1]$ and $\ES[F_2]$ are different minimal
foldings of $\ES[F_0]$. In particular, merging the three $d$-labelled
events as in $\ES[F_3]$ modifies the behaviour. In fact, in
$\ES[F_3]$, the event $d_{012}$ is not enabled in $C = \{ a \}$ since
$c \prec d_{012}$ and no event in $C$ is in conflict with $c$. Instead, in $\ES[F_0]$, the event $d_0$ is clearly enabled from $\{ a \}$.

Existence of a unique minimal folding could be possibly recovered by strengthening the notion of folding and, in particular, by requiring that foldings preserve and reflect histories. Note, however, that this would be against the spirit of our work where the notion of folding is not a choice. Rather, after having assumed hhp-bisimilarity as the reference behavioural equivalence, the notion of folding is essentially ``determined'' as a quotient (surjective function) that preserves the behaviour up to hhp-bisimilarity.

\section{Conclusions}
\label{se:conclusions}

We studied the problem of minimisation for poset event structures, a
class that encompasses many stable event structure models in the
literature, taking hereditary history preserving bisimilarity as
reference behavioural equivalence. We showed that a uniquely
determined minimal quotient always exists for poset event structures
and also in the subclass of prime event structures, while this is not
the case for various models extending prime event structures. We
showed that foldings between general poset event structures arise from foldings
of corresponding canonical prime event structures.
Finally, we provided a characterisation of foldings of prime event
structures, and discussed how this could be generalised to other
classes, developing explicitly the case of asymmetric event
structures.

As underlined throughout the paper, our theory of folding has many
connections with the literature on event structures. The idea of
``unfolding'' more expressive models to prime algebraic domains and
prime event structures has been studied by many
authors (e.g., in~\cite{Ren:PC,NPW:PNES,vGla:HPPG,GP:CSESPN,BC:SCCS}).
The same can be said for the idea of refining a single action into a complex computation (see, e.g.,~\cite{vGG:ENCC-AI} and references therein).
Instead, the problem of minimisation of event structures has received
less attention.
We already commented on the relation with the notion of abstraction
homomorphisms for {\pes}s~\cite{Cas:phd}, which captures the idea of
behaviour preserving abstraction in a subclass of structured {\pes}s.
In some cases, given a Petri net or an event structure
a special transition system can be extracted, on which minimisation is
performed.
In particular, in~\cite{MP:MTS} the authors propose an encoding of
safe Petri nets into causal automata, in a way that preserves
hp-bisimilarity. The causal automata can be transformed into a
standard labelled transition system, which in turn can be
minimised. However, in this way, the correspondence with the original
events is lost.

The notion of behaviour preserving function has been given an
elegant abstract characterisation in terms of open
maps~\cite{JNW:BFOM}. In the paper, we mentioned the
possibility of viewing our foldings as open maps and we observed that
various results admit a categorical interpretation. This gives clear
indications of the possibility of providing a general abstract view
of the results in this paper, something which represents an
interesting topic of future research.

The characterisation of foldings on prime (and asymmetric) event
structures can be used as a basis to develop, at least in the case of
finite structures, an algorithm for the definition of behaviour
preserving quotients. The fact that conditions for folding refer to
sets of events might make the minimisation procedure very
inefficient. Determining suitable heuristics for the identification of
folding sets and investigating the possibility of having more ``local''
conditions characterising foldings are interesting directions of
future development.

Although not explicitly discussed in the paper, considering elementary
foldings, i.e., foldings that just merge a single set of events, one
can indeed determine some more efficient folding rules. This is
essentially what is done for {\aes}s and {\fes}s
in~\cite{ABG:RESHPB}.  However, restricting to elementary foldings
is limitative, since it can be seen that general foldings
cannot be always decomposed in terms of elementary ones (e.g., it can be seen that in Fig.~\ref{fi:pes}, the folding $f_{02} : \ES[P_0] \to \ES[P_2]$ cannot be obtained as the composition of elementary foldings).

When dealing with possibly infinite event structures one could work on some finitary representation and try to devise reduction rules acting on the representation and inducing foldings on the corresponding event structure.
Observe that working, e.g., on finite safe Petri nets, the minimisation procedure would be necessarily incomplete, given that hhp-bisimilarity is known to be undecidable~\cite{JNS:UDG}.

\bibliography{bibliography}

\appendix

\section{Proofs for Section~\ref{se:preliminaries} (Event Structures and History Preserving Bisimilarity)}

\begin{alemma}[properties of histories]
  \label{le:hist}
  Let $\ES$ be an event structure. Then
  \begin{enumerate}
    
  \item \label{le:hist:1}
    for all $C \in \conf{\ES}$, we have $\hist{C}{x} \sqsubseteq C$,
    hence $\hist{C}{x} \in \conf{\ES}$;

  \item \label{le:hist:3}
    for all $C_1, C_2 \in \conf{\ES}$, $C_1 \sqsubseteq C_2$ iff
    for all $x \in C_1$, $\hist{C_1}{x} = \hist{C_2}{x}$;

  \item \label{le:hist:2}
    for all $H_1, H_2 \in \Hist{x}$, if $H_1 \cons H_2$ then $H_1 = H_2$;
    
  \end{enumerate}
\end{alemma}

\begin{proof}
  \begin{enumerate}

  \item Immediate by the definition of $\hist{C}{x}$.

  \item Let $C_1, C_2 \in \conf{\ES}$ such that $C_1 \sqsubseteq
    C_2$. For all $x \in C_1$ we have that
  \begin{align*}
    \hist{C_2}{x}
    & = \{ y \in C_2 \mid y \leq_{C_2} x \}\\
    & = \{ y \in C_1 \mid y \leq_{C_1} x \} & \mbox{[since $C_1 \sqsubseteq C_2$]}\\
    & = \hist{C_1}{x} 
  \end{align*}

  Conversely, assume that for all $x \in C_1$ we have that
  $\hist{C_1}{x} = \hist{C_2}{x}$.
  Then, since $x \in \hist{C_i}{x}$, for $i \in \{1,2\}$, clearly
  $C_1 \subseteq C_2$.
  Moreover, for all $y \in C_1$ and $x \in C_2$, if $x \leq_{C_2} y$
  then $x \in \hist{C_2}{y}$. Therefore, since by hypothesis
  $\hist{C_1}{y} = \hist{C_2}{y}$, we have $x \in C_1$ and
  $x \leq_{C_1} y$, as desired. Therefore, $C_1 \sqsubseteq C_2$.
  
\item Let $H_1, H_2 \in \Hist{x}$ and assume that $H_1 \cons H_2$. This
  means that there exists $C \in \conf{\ES}$ such that
  $H_1, H_2 \subseteq C$. Therefore, by point (\ref{le:hist:3}), we
  have $H_1 = \hist{H_1}{x} = \hist{C}{x} = \hist{H_2}{x} = H_2$.
\end{enumerate}
\end{proof}

 \begin{alemma}[configurations are reachable]
  \label{le:conf-reach}
  Let $\ES$ be an event structure and let $C \in \conf{\ES}$ be a
  configuration. Then $\emptyset \trans{}^* C$. More in detail,
  if $x_1, x_2, \ldots, x_n$ is any linearisation of $C$ compatible
  with $\leq_C$ then, for all $k \in \{1, \ldots, n\}$,
  $\{ x_1, \ldots, x_{k-1}\} \trans{x_k} \{ x_1, \ldots, x_{k-1},
  x_k\}$ .
\end{alemma}

\begin{proof}
  Immediate consequence of the prefix-closedness of
  the family of configurations.
\end{proof}

\section{Proofs for Section~\ref{se:folding} (Foldings of Event Structures)}

\begin{alemma}[foldings are closed under composition]
  \label{le:folding-composition}
  Let $\ES$, $\ES'$, $\ES''$ be event structures and let
  $f : \ES \to \ES'$ and $f' : \ES' \to \ES''$ be foldings. Then
  $f' \circ f : \ES \to \ES''$ is a folding.
\end{alemma}
\begin{proof}
  We rely on the characterisation of foldings provided in
  \cref{le:morph-to-fold}. Let $C_1 \in \conf{\ES}$ and assume that
  $f'(f(C_1)) \trans{x''} C_2''$. Since $f(C_1) \in \conf{\ES'}$ and
  $f'$ is a folding, there exists $x'$ such that
  $f(C_1) \trans{x'} C_2'$ with $f'(x') = x''$ and $f'(C_2') =
  C_2''$. In turn, since $f$ is a folding, from
  $f(C_1) \trans{x'} C_2'$, we derive the existence of a transition
  $C_1 \trans{x} C_2$ with $f(x) = x'$ and $f(C_2) = C_2'$. Therefore
  $f'(f(x)) = x''$ and $f'(f(C_2)) = C_2''$, as desired.
\end{proof}

\lemorphtofold*

\begin{proof}
  We have to show that
  $R_f = \{ (C, f_{|C}, f(C)) \mid C \in \conf{\ES}\}$ satisfies
  conditions (1) and (2) of \cref{de:hp-bisim}. Condition
  (2) is in the hypotheses. Concerning (1), let $C_1 \in \conf{\ES}$ and
  consider a transition $C_1 \trans{x} C_2$. Then by definition of
  morphism, $f(C_i) \in \conf{\ES}$ and isomorphic to $C_i$, for
  $i \in \{1,2\}$. Therefore $f(C_1) \trans{f(x)} f(C_2)$.
\end{proof}

Relying on \cref{le:morph-to-fold} we can derive that foldings
arise as open maps in the sense of~\cite{JNW:BFOM}.

\begin{adefinition}[open map]
  \label{de:open-map}
  Let $\cat{M}$ be a category and let $\cat{C}$ be a subcategory of
  $\cat{M}$. A morphism $f: M \to M'$ is $\cat{C}$-\emph{open} if for
  all morphisms $e: C \to C'$ and commuting square
  \begin{center}
    \begin{tikzcd}%
      C \arrow[r, "c" above] \arrow[d, "e" left] & \ES  \arrow[d, "f" right]\\
      C' \arrow[r, "c'" below]  \arrow[ru, dotted, "c''"] & \ES'
      & {}
    \end{tikzcd}
  \end{center}
  there exists a morphism $c'' : C' \to \ES$ such that the two
  triangles commute.  
\end{adefinition}

Let $\cat{Pom}$ denote the subcategory of $\cat{ES}$ having
conflict-free {\pes}s as objects and injective morphisms as
arrows. Then \cref{le:morph-to-fold} leads to show that foldings are
$\cat{Pom}$-open morphisms in $\cat{ES}$, generalising to our setting
a result proved for prime event structures in~\cite{JNW:BFOM}.

\begin{alemma}[foldings as open maps]
  \label{le:folding-open}
  Let $\ES$, $\ES'$ be event structures and let $f : \ES \to \ES'$ be
  a morphism. Then $f$ is a folding if and only if $f$ is $\cat{Pom}$-open.
\end{alemma}

\begin{proof}
  Let $f$ be a folding. In order to prove that $f$ is a
  $\cat{Pom}$-open map, assume to have a commuting square as in
  \cref{de:open-map}. Since $C$ is a conflict-free prime
  event structures, its set of events, ordered by causality, which
  abusing the notation, we still denote by $C$ is a
  configuration. Since $c$ is a morphism $c(C) \in \conf{\ES}$ and
  $c(C) \simeq C$, and thus $f(c(C)) \in \conf{\ES'}$ and
  $f(c(C)) \simeq C$. Similarly, $c'(C') \in \conf{\ES'}$ and
  $c'(C') \simeq C'$.  Finally observe that $e(C) \sqsubseteq
  C'$. Thus $c'(e(C)) = f(c(C)) \sqsubseteq c'(C')$, meaning that
  $f(c(C)) \trans{X'} c'(C')$ for a suitable $X'$. By definition of
  folding, there must be a transition $c(C) \trans{X} D$ such that
  $f(D) = c'(C')$. Therefore, we can define $c'' : C' \to \ES$ as
  follows: for all $x' \in C'$, let $c''(x')$ be the unique $y \in D$
  such that $f(y) = c'(x')$.
  
  Conversely, assume that $f$ is an $\cat{Pom}$-open map. We show that
  $f$ satisfies the condition of \cref{le:morph-to-fold}. Let
  $C_1 \in \conf{\ES}$ and consider a transition
  $f(C_1) \trans{x'} C_2'$. If we view configurations $C_1, C_2'$ as
  pomsets, then we can build the following commuting square
  \begin{center}
      \begin{tikzcd}%
      C_1 \arrow[hookrightarrow, r] \arrow[d, "f_{|C_1}" left] & \ES  \arrow[d, "f" right]\\
      C_2' \arrow[hookrightarrow, r]  \arrow[ru, dotted, "c''"] & \ES'
      & {}
    \end{tikzcd}
  \end{center}
  By the fact that $f$ is open, we get the morphism $c''$, and it is
  immediate to see that $C_1 \trans{x} c''(C_2')$ is the desired
  transition that completes the proof.
\end{proof}

\lefoldisequivalence*

\begin{proof}
  Consider the function $g : \quotient{E}{\equiv_f} \to \ES'$ defined by
  $g(\eqclass{x}{\equiv_f}) = f(x)$. It is well defined, since all
  elements in $\eqclass{x}{\equiv_f}$ have the same $f$-image, and
  clearly injective. Moreover, it is also surjective. In fact, if
  $x' \in \ES'$ then there exists $C' \in \conf{\ES'}$ such that
  $x' \in C'$. By \cref{le:conf-reach}, configuration $C'$
  is reachable from the empty one, and thus, since $f$ is an
  hp-bisimulation, there exists $C \in \conf{\ES}$ such that
  $C' = f(C)$. Therefore there is $e \in C$ such that $f(x)=x'$ and
  thus $g(\eqclass{x}{\equiv_f})=x'$.

  Finally, observe that by definition, for all configuration
  $C' \in \conf{\quotient{\ES}{\equiv_f}}$, we have $g(C') \simeq C'$,
  hence we conclude.
\end{proof}

\begin{alemma}[factorising morphisms]
  \label{le:factor-morph}
  Let $\ES$, $\ES'$, $\ES''$ be event structures and let
  $f : \ES'' \to \ES'$ be a morphism and $h : \ES'' \to \ES$ be a folding.
  Let
  $g : \ES \to \ES'$ be a function such that $f = g \circ h$.
  \begin{center}
    \begin{tikzcd}%
      \ES \arrow[r, "g" above]
      &
      \ES'\\
      \ES'' \arrow[ur, "f" below] \arrow[u, "h" left]
      & {}
    \end{tikzcd}
  \end{center}
  Then $g$ is a morphism. Moreover, if $f$ is a folding then $g$ is. 
\end{alemma}

\begin{proof}
  Let us show that $g$ is a morphism. For all $C \in \conf{\ES}$, since
  $h$ is a folding, there exists $C'' \in \conf{\ES''}$ such that
  $h(C'') = C$ and $C'' \simeq C$.
  Since $f$ is a morphism $f(C'') \in \conf{\ES'}$. Therefore
  $g(C) = g(h(C'')) = f(C'')$, as desired.

  Let assume now that $g$ is a folding. Let $C_1 \in \conf{\ES}$ and
  suppose that there is a transition $g(C_1) \trans{x'} C_2'$. Since
  $h$ is a folding, there is a configuration $C_1'' \in \conf{\ES''}$
  such that $C_1 = h(C_1'')$. Therefore
  $f(C_1'') = g(h(C_1'') = g(C_1) \trans{x'} C_2'$.  Since $f$ is a
  folding there is a transition $C_1'' \trans{x''} C_2''$ with
  $f(C_2'') = C_2'$. Therefore
  $h(C_2'') = C_1 \trans{h(x'')} h(C_2'')$ with
  $g(h(C_2'')) = f(C_2) = C_2'$, as desired.
\end{proof}

\prfoldjoin*

\begin{proof}
  We actually show that the construction described in the statement
  produces the pushout in the category $\cat{ES}$ and also in
  $\cat{ES_f}$. Consider the diagram
  \begin{center}
    \begin{tikzcd}[row sep=small]      
      &
      \ES \arrow[dl, "f'" above] \arrow[dr, "f''" above]
      &
      \\
      \ES' \arrow[dr, "g'" below]
      &
      {}
      &
      \ES''  \arrow[dl, "g''" below]\\
      &
      \ES'''
      &
    \end{tikzcd}
  \end{center}

  Observe that $\ES'''$, with functions $g'$ and $g''$ is the pushout
  in $\cat{Set}$, as it easily follows recalling that $f'$ and $f''$ are
  surjective. Another immediate observation is that the set of
  configurations of $\ES'''$ can be written
  \begin{equation}
    \label{eq:conf-po}
    \conf{\ES'''} = \{ g'(f'(C)) \mid C \in \conf{\ES} \} = \{
    g''(f''(C)) \mid C \in \conf{\ES} \}
  \end{equation}

  We prove that $g'$ is a folding. In fact
  \begin{itemize}
    
  \item $g'$ is a morphism.\\
    For all $C' \in \conf{\ES'}$, since $f'$ is a folding, there is
    $C \in \conf{\ES}$ such that $f'(C) = C'$. Therefore
    $g'(C') = g'(f'(C)) \in \conf{\ES'''}$, by construction. Moreover,
    $g'$ is injective on $C'$. In fact, take $x', y' \in C'$, with
    $g'(x') = g'(y')$. Since $C' = f'(C)$, there are $x, y \in C$ such
    that $f'(x)=x'$ and $f'(y)=y'$. Therefore, $g'(f'(x))=g'(f'(y))$,
    and thus, by the properties of pushouts, $f''(x) = f''(y)$. Since
    $f''$ is a folding, thus a morphism, this implies $x=y$ and thus
    $x' = f'(x)= f'(y)=y'$, as desired.
    
  \item $g'$ is a folding.\\
    Let $C_1' \in \conf{\ES'}$ and assume that
    $f'(C_1') \trans{x'''} D_2'''$. By (\ref{eq:conf-po}) we know that
    there is $D_2 \in \conf{\ES}$ such that $D_2''' = g'(f'(D_2))$ and
    $D_2 \simeq D_2'''$. Therefore, there is $D_1 \sqsubseteq D_2$
    such that $f'(g'(D_1)) = g'(C_1')$ and
    \begin{equation}
      \label{eq:step-po}
      D_1 \trans{x} D_2.
    \end{equation}
    
    Define $D_1' = f'(D_1) \in \conf{\ES'}$.
    Now, since $f'$ is a folding and $C_1' \in \conf{\ES_1}$, there is
    also $C_1 \in \conf{\ES}$ such that $f'(C_1) = C_1'$. Recall that
    $g'(D_1') = f'(g'(D_1)) = g'(C_1')$, hence, by pushout properties,
    it must be $f''(C_1) = f''(D_1)$.
    From (\ref{eq:step-po}), since $f''$ is a folding, we deduce
    $f''(C_1) = f''(D_1) \trans{x''} D_2''$, with $f''(D_2) = D_2''$.
    And, using again the fact that $f''$ is a folding, this implies
    $C_1 \trans{y} C_2$, with $f''(C_2) = D_2'' = f''(D_2)$.

    Now, we use the fact that $f'$ is a folding, and derive that
    $C_1' = f'(C_1) \trans{f'(x_1)} f'(C_2)$. If we call
    $C_2' = f'(C_2)$, we have that $g'(C_2') = g'(D_2')$, as desired,
    since $f''(C_2) = f''(D_2)$.
  \end{itemize}
  In the same way, one concludes that also $g''$ is a folding.

  Given any other $\ES_1$ with morphisms $g_1' : \ES' \to \ES_1$ and
  $g_1'' : \ES'' \to \ES_1$ such that
  $g_1' \circ f' = g_2' \circ f''$, we show that there exists a unique
  morphism $h : \ES''' \to \ES_1$ that makes the diagram commute.
    
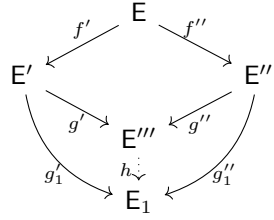
\begin{center}
    \begin{tikzcd}[row sep=small]      
      &
      \ES \arrow[dl, "f'" above] \arrow[dr, "f''" above]
      &
      \\
      \ES' \arrow[dr, "g'" below] \arrow[ddr, bend right, "g_1'" below]
      &
      {}
      &
      \ES''  \arrow[dl, "g''" below] \arrow[ddl, bend left, "g_1''" below]\\
      &
      \ES''' \arrow[d, dotted, "h" left]
      &\\
      &
      \ES_1
      &
    \end{tikzcd}
  \end{center}
  Consider the unique map $h : \ES''' \to \ES_1$ making the diagram
  commute in $\cat{Set}$. Since $g'$ is a folding and $g_1'$ is a
  morphism, by \cref{le:factor-morph}, also $h$ is a morphism.
  This proves that $\ES'''$ is a pushout in $\cat{ES}$.
  
  By the same result, if $g_1'$ is a folding, also the mediating
  morphism $h$ is. This means that the same construction produces a
  pushout in $\cat{ES_f}$.
\end{proof}

As a counterexample to the existence of  pushouts in $\cat{ES}$ for general morphisms, consider the obvious mappings $f_{45} : \ES[P_4] \to \ES[P_5]$ and $f_{46} : \ES[P_4] \to \ES[P_6]$ in Fig.~\ref{fi:no-pushout}.
  
\begin{figure}
  \begin{center}
  \setlength{\tabcolsep}{24pt}
  \begin{tabular}{ccc}
    \begin{tikzcd}[boxedcd, sep=4mm]
      a_1
      \arrow[d] \arrow[conflict, r]
      & a_2 \arrow[d] \\
      b_1           & b_2 
    \end{tikzcd}
    &
    \begin{tikzcd}[boxedcd, sep=4mm]
      a_1
      \arrow[d]
      & a_2 \arrow[d] \\
      b_1           & b_2 
    \end{tikzcd}
    &
    \begin{tikzcd}[boxedcd, row sep=4mm, column sep=3mm]
      a_{12} \arrow[d] \\
      b_{12}
    \end{tikzcd}
    \\
    $\ES[P_4]$ & $\ES[P_5]$ & $\ES[P_6]$
  \end{tabular}
  \end{center}
  \caption{Non existence of pushout of general morphisms}
  \label{fi:no-pushout}
\end{figure}
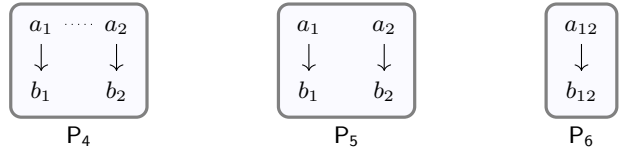

\begin{alemma}[multi-colimit]
  \label{le:multi}
  Let $\ES$ be an event structure. Each collection of foldings
  $f_i : \ES \to \ES_i$ with $i \in I$ has a colimit in
  $\cat{ES}$. Therefore the coslice
  category $(\ES \downarrow \cat{ES_f})$ has a terminal object.
\end{alemma}

\begin{proof}
  When $I$ is finite, the proof proceeds by straightforward induction
  on $I$, using \cref{pr:fold-join}.
  If instead $I$ is infinite, let $\ES'$ be the colimit of the
  $f_i$'s in $\cat{Set}$.
  
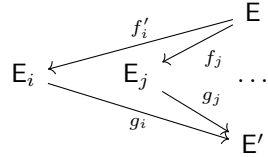
\begin{center}
    \begin{tikzcd}[row sep=small]      
      & &
      \ES \arrow[dll, "f_i'" above] \arrow[dl, "f_j" below right]
      &
      \\
      \ES_i \arrow[drr, "g_i" below]
      &
      \ES_j  \arrow[dr, "g_j" above right]
      &
      \ldots
      \\
      &
      &
      \ES'
      &
    \end{tikzcd}
  \end{center}
  with configurations
  $\conf{\ES'} = \{ g_i(f_i(C)) \mid C \in \conf{\ES} \}$. The proof
  of the fact that the $g_i$'s are foldings then proceeds as in
  \cref{pr:fold-join}.
  The only delicate point is the following. Given configurations
  $C, C' \in \conf{\ES}$, define $C_i = f(C)$ and
  $C_i'= f(C') \in \conf{\ES_i}$. If $g_i(C_i) = g_i(C_i')$, then it
  is not necessarily the case that $f_j(C) = f_j(C')$ for some
  $j \in I$. However, since configurations are finite, there is a
  finite subset $J \subseteq I$ such that, if $\ES_J$ is the colimit
  of $\{ f_j \mid j \in J \}$ and $f_J : \ES \to \ES_J$ the
  corresponding folding, whose existence is proved in the first part,
  then $f_J(C) = f_J(C')$.
  Exploiting this fact, we can conclude exactly as in \cref{pr:fold-join}.
\end{proof}

\prfoldinglattice*

\begin{proof}
  Immediate corollary of Lemma~\ref{le:multi}.
\end{proof}

\begin{alemma}[hhp-bisimulation as an event structure]
  \label{le:bisim-as-es}
  Let $\ES'$, $\ES''$ be event structures and let $R$ be a
  hhp-bisimulation between them. Then there exists a (prime) event
  structure $\ES_R$ and two foldings $\pi' : \ES_R \to \ES'$ and
  $\pi'' : \ES_R \to \ES''$.
\end{alemma}

\begin{proof}
  Let $\ES'$, $\ES''$ be event structures and let $R$ be a
  hhp-bisimulation between them. Define $\ES_R$ as follows. Events are
  histories related by $R$, namely the triples
  $\{ (H', f, H'') \} \mid H' \in \Hist{\ES'} \}$, labelled by
  $\lambda_{\ES_R}(H', f, H'') = \lambda_{\ES}(x')$ when
  $H' \in \Hist{x'}$.
  For each $(C', f, C'') \in R$, define
  \begin{center}
    $C_f = \{ (\hist{C'}{x'} , f_{|\hist{C'}{x'}}, \hist{C''}{f(x')})
    \mid x \in C' \}$
  \end{center}
  ordered by pointwise inclusion, i.e.,
  $(H'_1, f_1, H''_1) \leq_{C_f} (H'_2, f_2, H''_2)$ if
  $f_1 \subseteq f_2$, and thus $H'_1 \subseteq H''_1$,
  $H'_2 \subseteq H''_2$.
  The set of configurations of $\ES_R$  is
  $\conf{\ES_R} = \{ C_R \mid C \in \conf{\ES}\}$.

  It is easy to see that $\conf{\ES_R}$ is
  well-defined. Prefix-closedness of $\conf{\ES_R}$ follows from the
  fact that $R$ is downward-closed by definition of
  hhp-bisimulation. It can be seen that $\ES_R$ is actually a prime
  event structure, with causality defined by
  $(H'_1, f_1, H''_1) \leq (H'_2, f_2, H''_2)$ if
  $H'_1 \sqsubseteq H'_2$ and $f_1 \sqsubseteq f_2$, and conflict
  defined by $(H'_1, f_1, H''_1) \# (H'_2, f_2, H''_2)$ if there is no
  $(C',f, C'') \in R$ such that $H'_1, H'_2 \sqsubseteq C'$ and
  $f_1, f_2 \subseteq f$.

  Consider two configurations $C_{f_1}, C_{f_2} \in \conf{\ES_R}$,
  arising from the triples $(C'_i, f_i, C''_i) \in R$, for
  $i \in \{1,2\}$. Then it holds that
  \begin{quote}
    $C_{f_1} \sqsubseteq C_{f_2}$\\
    iff $C_{f_1} \subseteq C_{f_2}$\\
    iff for all $x' \in C'_1$, $(\hist{C_1'}{x'} , {f_1}_{|\hist{C_1'}{x'}}, \hist{C_1''}{f_1(x')}) \in C_{f_2}$\\
    iff for all $x' \in C'_1$, $\hist{C_1'}{x'}=\hist{C_2'}{x'}$ and $f_1(x') = f_2(x')$\\
    iff $C_1' \sqsubseteq C_2'$ and $f_1 \subseteq f_2$.
  \end{quote}

  We can now define $\pi' : \ES_R \to \ES'$ as $\pi'(H',f,H'') = x'$
  if $H' \in \Hist(x')$ and, similarly, $\pi'' : \ES_R \to \ES'$ as
  $\pi''(H',f,H'') = x''$ if $H'' \in \Hist(x'')$.

  Then $\pi'$ and $\pi''$ are well-defined morphisms and they are
  foldings. We prove this for $\pi'$ (for $\pi''$ the proof is
  completely analogous).

  \begin{itemize}
    
  \item $\pi'$ is a morphism.\\
    This is immediate by observing that for any configuration
    $C_f \in \conf{\ES_R}$, arising from the triple
    $(C', f, C'') \in R$, then we have $\pi'(C_f) = C'$. Note that,
    concerning the local order, for $x', y' \in C'$ we have
    $(\hist{C'}{x'} , f_{|\hist{C'}{x'}}, \hist{C''}{f(x')})
    \leq_{C_f} (\hist{C'}{y'} , f_{|\hist{C'}{y'}},
    \hist{C''}{f(y')})$ iff inclusion holds pointwise iff
    $x' \in \hist{C'}{y'}$ iff $x' \leq_{C'} y'$, which means
    $\pi'(\hist{C'}{x'}) = x' \leq_{C'} y' =\pi'(\hist{C'}{y'})$.
    
  \item $\pi$ is a folding.\\
    In fact, for any configuration $C_f \in \conf{\ES_R}$, arising
    from the triple $(C', f, C'') \in R$, if
    $\pi'(C_f) = C' \trans{x'} D'$ then, since $R$ is an
    hhp-bisimulation, there is $C'' \trans{x''} D''$ with
    $(C'', g, D'') \in R$ with $g = f[x' \mapsto x'']$. Hence, if we
    let $H' = \hist{D'}{x'}$, we have that
    $C_f \trans{(H', g_{|H'}, g(H'))} C_g$ and $\pi'(C_g) = D'$, as
    desired.
  \end{itemize}
\end{proof}

\prfoldsubsume*

\begin{proof}
  Let $R$ be a hhp-bisimulation between $\ES$ and
  $\quotient{\ES}{\equiv_f}$. Consider the event structure $\ES_R$ and
  the foldings $\pi : \ES_R \to \ES$ and
  $\pi' : \ES_R \to \quotient{\ES}{\equiv_f}$, given by
  \cref{le:bisim-as-es}.
  By \cref{pr:fold-join} we can close the diagram as follows:
    
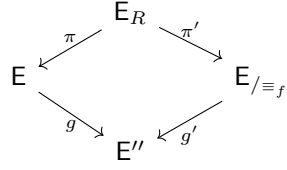
\begin{center}
    \begin{tikzcd}[row sep=small]      
      &
      \ES_R \arrow[dl, "\pi" above] \arrow[dr, "\pi'" above]
      &
      \\
      \ES \arrow[dr, "g" below]
      &
      {}
      &
      \quotient{\ES}{\equiv_f}  \arrow[dl, "g'" below]\\
      &
      \ES''
      &
    \end{tikzcd}
  \end{center}
  and both $g$ and $g'$ are foldings. Then
  $\ES'' = \quotient{\ES}{\equiv_g} =
  \quotient{\left(\quotient{\ES}{\equiv_f}\right)}{\equiv_{g'}}$ and we conclude.
\end{proof}

\begin{alemma}[configurations of the canonical {\pes}]
  \label{le:conf-in-pes}
  Let $\ES$ be an event structure.  Then $\conf{\ES}$ and
  $\conf{\pr(ES)}$ seen as partial orders, ordered by prefix, are
  isomorphic.

  More in detail, for all $C \in \conf{\ES}$ it holds
  $\hset{C} = \{ \hist{C}{x} \mid x \in C \}$, with inclusion as
  local order, is in $\conf{\pr(\ES)}$. Moreover $C \simeq \hset{C}$
  and $\hset{\cdot} : \conf{\ES} \to \conf{\pr(\ES)}$ is a poset
  isomorphism.

  Its inverse is as follows. For $D \in \conf{\pr(\ES)}$
  consider $\flt{D} = \bigcup D$. Then, for each $x \in \flt{D}$ there
  exists a unique $H_x \in D$ such that $H_x \in \Hist{x}$. Define the
  order $\leq_{\flt{D}}$, for $x, y \in {\flt{D}}$, by
  $x \leq_{\flt{D}} y$ iff $x \in H_y$. Then $\flt{D} \in \conf{\ES}$
  and $\flt{D} \simeq D$ as posets.
\end{alemma}

\begin{proof}
  Let $C \in \conf{\ES}$ and let us show that
  $\hset{C} = \{ \hist{C}{x} \mid x \in C \}$, with inclusion as local
  order, is in $\conf{\pr(\ES)}$. First, note that $\hset{C}$ is
  consistent by construction, since $\hist{C}{x} \sqsubseteq C$ for
  all $x \in C$. Moreover, it is causally closed. In fact, if
  $H \sqsubseteq \hist{C}{x}$ for some $H \in \Hist{\ES}$, then, if
  $H \in \Hist{y}$, by \cref{le:hist}(\ref{le:hist:3}) we have
  $H = \hist{\hist{C}{x}}{y} = \hist{C}{y} \in \hset{C}$ .
  Moreover, $\hset{C}$ is isomorphic to $C$, the isomorphism established
  by the mapping $\hist{C}{x} \mapsto x$. It is clearly
  bijective. Moreover, for all $x_1, x_2 \in C$ it holds that
  $\hist{C}{x_1} \subseteq \hist{C}{x_2}$ iff $x_1 \in \hist{C}{x_2}$
  and thus $x_1 \leq_C x_2$.

  Let us show that $\hset{\cdot} : \conf{\ES} \to \conf{\pr(\ES)}$ is
  a poset isomorphism. It is injective. In fact, if
  $\hset{C_1} = \hset{C_2}$ then clearly $C_1$ and $C_2$ contain the
  same events. Moreover, $\leq_{C_1} = \leq_{C_2}$ and thus the two
  configurations coincide. Otherwise, there would be $x, y \in C_1$
  such that $x \leq_{C_1} y$ and $\neg (x \leq_{C_2} y)$, or
  conversely $\neg(x \leq_{C_1} y)$ and $x \leq_{C_2} y$. Assume,
  without loss of generality, that we are in the first case. Then
  $x \in \hist{C_1}{y}$ and $x \not\in \hist{C_2}{y}$, and thus
  $\hset{C_1} \neq \hset{C_2}$ contradicting the hypotheses.
  Moreover, it preserves and reflects the prefix order, i.e., given
  $C_1, C_2 \in \conf{\ES}$ we have $C_1 \sqsubseteq C_2$ iff
  $\hset{C_1} \subseteq \hset{C_2}$ as it immediately follows from
  \cref{le:hist}(\ref{le:hist:3}).
  
  We conclude, by showing that it is also surjective.
  Consider any configuration $D \in \conf{\pr(\ES)}$. Since $D$ has no
  conflicts, its elements are pairwise compatible. Therefore, by
  coherence of the class of configurations, there exists
  $C \in \conf{E}$ such that $H \sqsubseteq C$ for all $H \in D$. Let
  $\flt{D} = \bigcup D$. Then, for each $x \in \flt{D}$ there exists a
  unique $H_x \in D$ such that $H_x \in \Hist{x}$, since by
  \cref{le:hist}(\ref{le:hist:3}) different histories of the same
  event are not compatible. Define the order $\leq_{\flt{D}}$, for
  $x, y \in {\flt{D}}$, by $x \leq_{\flt{D}} y$ iff $x \in H_y$. It is
  easy to check that $\flt{D} \sqsubseteq C$, and thus by prefix
  closedness of $\conf{\ES}$, we have $\flt{D} \in \conf{\ES}$. It is
  now immediate to see that $\hset{\flt{D}} = D$, thus we conclude.  
\end{proof}

\leunfespes*

\begin{proof}
  This results can be derived from the characterisation of foldings as
  $\cat{Pom}$-open maps (\cref{le:folding-open}), the fact that
  $\cat{PES}$ is a coreflective subcategory of $\cat{ES}$
  (\cref{le:factor-fold}) and then using~\cite[Lemma 6(ii)]{JNW:BFOM}.
  
  Explicitly, the fact that $\phi_{\ES}$ is a morphism immediately
  follows from the observation that $\phi_{\ES}(D) = \flt{D}$. Then by
  \cref{le:conf-in-pes}, we have $D \simeq \phi_{\ES}(D)$, as desired.

  In order to conclude that it is a folding we show that given
  $D_1 \in \conf{\pr(\ES)}$, if $\phi_{\ES}(D_1) \trans{x} C_2$ then
  $D_1 \trans{H} D_2$ with $\phi_{\ES}(D_2) = C_2$.
  Let $C_1 = \phi_{\ES}(D_1)$ and assume $C_1 \trans{x} C_2$. By
  definition of transition (\cref{de:transition}), we have
  $C_1 \sqsubseteq C_2$.
  Let $H_x = \hist{C_2}{x}$. By
  definition of $\pr(\ES)$, the causes
  $\causes{H_x} = \{ \hist{H_x}{y} \mid y \in H_x \}$.
  For all $y \in H_x \setminus \{x\}$, clearly $y \in C_1$. Moreover
  $\hist{H_x}{y} = \hist{C_2}{y} = \hist{C_1}{y}$. Therefore, by
  \cref{le:hist}(\ref{le:hist:3}), $\hist{H_x}{y} \in
  D_1$.
  We thus conclude that
  \begin{center}
    $D_1 \trans{H_x} D_2$
  \end{center}
  and moreover $\phi_{\ES}(D_2) \simeq C_2$. For the last statement, the
  only thing to observe is that the image of the causes of $H_x$ are
  exactly the causes of $x$. Indeed we have, for all $H \in D_2$, say
  $H \in \Hist{y}$, that $H \sqsubseteq H_x$ iff $y \in H_x$ iff
  $y \leq_{C_2} x$, as desired.
\end{proof}

\lefactorfold*

\begin{proof}
  The function $g$ can be defined, for all $x' \in \ES[P']$ as
  \begin{center}
    $g(x') = f(\causes{x'})$
  \end{center}
  Note that this is a well-defined morphism. First observe that
  $g(x') \in \Hist{\ES}$, hence it is an event in $\pr(\ES)$. In fact,
  for all $x' \in \ES[P']$, since $f$ is a morphism and
  $\causes{x'} \in \conf{\ES[P']}$, $f(\causes{x'}) \in \conf{\ES}$, and
  $f(\causes{x'}) \simeq \causes{x'}$, therefore
  $g(x') = f(\causes{x'}) = \hist{f(\causes{x'})}{f(x')} \in \Hist{\ES}$.
  Moreover, the reasoning above shows that $g(x') \in
  \Hist{f(x')}$. Therefore, if $g(x') = g(y')$ then $f(x') = f(y')$. This
  fact, recalling that $f$ is injective on configurations, implies
  that also $g$ is.
  Finally, for all $C' \in \conf{\ES[P']}$, since $f$ is a morphism,
  $f(C') \in \conf{\ES}$ and $f(C') \simeq C'$. Therefore its
  $g$-image is
  \begin{align*}
    g(C') & =  \{ g(x') \mid x' \in C' \}\\
          & =  \{ f(\causes{x'}) \mid x' \in C' \}\\
          & =  \{ \hist{f(\causes{x'})}{f(x')} \mid x' \in C' \}
          & \mbox{[Since morphisms preserve prefix order]}\\
          & =  \{ \hist{f(C')}{f(x')} \mid x' \in C' \}\\
          & = \hset{f(C')}
  \end{align*}
  Hence, by \cref{le:conf-in-pes},
  $g(C') = \hset{f(C')} \in \conf{\pr(\ES)}$ and
  $\hset{f(C')} \simeq C'$, as desired.

  For the second part, assume that $f$ is a folding and let us show
  that also $g$ is. We use the characterisation in
  \cref{le:morph-to-fold}. Let $C_1' \in \conf{\ES[P']}$ and
  assume that $g(C_1') \trans{H} D_2$. Since $\phi_{\ES}$ is a
  morphism,  this implies that
  $f(C_1') = \phi_{\ES}(g(C_1')) \trans{\phi_{\ES}(H)} \phi_{\ES}(D_2)$.
  Since $f$ is a folding, by \cref{le:morph-to-fold}, there exists
  a transition $C_1' \trans{x'} C_2'$ such that
  $f(C_2') = \phi_{\ES}(C_2)$.
  Observe that this implies $f(x') = \phi_{\ES}(H)$ and more generally
  $f(\causes{x'}) = \phi_{\ES}(\causes{H})$, but since
  $\phi_{\ES}(\causes{H}) = H$
  \begin{center}
    $f(\causes{x'}) =H$.
  \end{center}

  We only need to show that $g(C_2') = D_2$. This is an immediate
  consequence of the fact that
  $g(C_2') = g(C_1') \cup \{ g(x') \} = D_1 \cup \{ H \} = D_2$, as
  desired.
\end{proof}

\prfoldespes*

\begin{proof}
  This can be derived from the characterisation of foldings as
  $\cat{Pom}$-open maps (\cref{le:folding-open}), the fact that
  $\cat{PES}$ is a coreflective subcategory of $\cat{ES}$
  (\cref{le:factor-fold}) and then using~\cite[Lemma 6(iii)]{JNW:BFOM}.

  Explicitly, let $\ES, \ES'$ be event structures, let $f : \ES \to \ES'$ be a
  morphism and consider the commuting diagram
  \begin{center}
    \begin{tikzcd}%
      \ES \arrow[r, "f" above] 
      &
      \ES'\\
      \pr(\ES) \arrow[r, "\pr(f)" below] \arrow[u, "\phi_{\ES}" left]
      &
      \pr(\ES') \arrow[u, "\phi_{\ES'}" right]
    \end{tikzcd}
  \end{center}
  If $f$ is a folding then $f \circ \phi_{\ES} : \pr(\ES) \to \ES'$ is
  a composition of foldings and thus, by \cref{le:fold-comp}, it
  is a folding. In turn, by \cref{le:factor-fold} this implies
  that $\pr(f)$ is a folding.

  Conversely, if $\pr(f)$ is a folding, then
  $\phi_{\ES} \circ \pr(f) : \pr(\ES) \to \ES'$ is a composition of
  foldings and thus, by \cref{le:fold-comp}, it is a folding.  In turn, by
  \cref{le:factor-morph} this implies that $f$ is a folding.
\end{proof}

\section{Some Properties of  Morphisms and Foldings}
\label{se:abstract-properties}

In this section, we define some relations between the events of an event structure, based on the way in which such events occur in configurations. They can be used to prove general properties of morphisms
of event
structures, that then can be instantiated on specific subclasses.

\begin{adefinition}[precedence]
  \label{de:precedence}
  Let $\ES$ be an event structure.  The \emph{precedence} as
  the relation $\prc \subseteq \ES \times \ES$, defined for $x, y \in \ES$ by
  $x \prc y$ if for all $C \in \conf{\ES}$ such that
  $x, y \in C$ it holds $x <_C y$.
  We say that $\ES$ has \emph{global
  precedence} if for $x, y \in \ES$, if $x, y \in C$ and $x <_C y$ then
  $x \prc y$.
\end{adefinition}

In words, $x \prc y$ whenever in each computation where $x, y$ occur necessarily $x$ occurs before $y$. The precedence relation is useful also to define a notion of semantic conflict.
Observe that for any configuration $C$ the precedences expressed by
$\prc$ are always respected by $\leq_C$, i,.e.,
$\prc_{C}^* \subseteq \leq_C$.  When the event structure has global
precedence, the precedence relation is sufficient to completely
characterise the local order of configuration, i.e., for all
configurations $C$ it holds that $<_C = (\prc_{|C})^*$.

Closely connected, we can introduce a notion of semantic conflict.

\begin{adefinition}[conflict]
  \label{de:conflict}
  Let $\ES$ be an event structure.  The \emph{conflict} is
  relation $\# \subseteq 2^E$, defined for a finite $X \subseteq E$ by
  $\# X$ if there is no $C \in \conf{\ES}$ such that
  $X \subseteq C$. When $\{x,y\}$ we often write $x \# y$.
\end{adefinition}

We observe that conflict and precedence are strictly related. In particular, binary conflict can be characterised in terms of precedence.

\begin{aproposition}[precedence vs conflict]
  Let $\ES$ be an event structure. Then
  \begin{itemize}

  \item 
    for $X \subseteq E$, if $\prc_{|X}$ is cyclic then $\# X$.
    
  \item for $x, y \in E$, we have $x \# y$ iff $x \prec y \prec x$.
\end{itemize}
\end{aproposition}

\begin{proof}
  \begin{itemize}
  \item Let $X \subseteq E$. If $\prc_{|X}$ is cyclic, i.e., there are
    $x_1, \ldots, x_n \in X$ such that
    $x_1\prc x_2 \prc \ldots x_n \prc x_1$ then the events
    $x_1, \ldots, x_n$ and thus $X$ can never occur together in the same
    computation, i.e., there cannot be $C \in \conf{\ES}$ such
    that $X \subseteq C$. In fact, otherwise, we should have
    $\prc_{|C}^* \subseteq \leq_C$, contradicting the fact that $\leq_C$
    is a partial order. In words, each of the events $x_i$ should occur
    before the others, which is impossible.
  \item In particular, if $x \# y$ then $x, y$ can never be in the
    same computation, hence trivially $x \prec y$ and $y \prec x$, and
    observe that also the converse holds.
  \end{itemize}
\end{proof}

Morphism on event structures can be shown to enjoy interesting properties with respect to the semantic relations.

\begin{alemma}[morphism properties]
  \label{le:morphism-properties}
  Let $\ES, \ES'$ be event structures and let
  $f : \ES \to \ES'$ be a morphism. Then for all
  $x, y \in E$
  \begin{enumerate}
    
  \item if $f(x) \prc f(y)$ then $x \prc y$;
  \item if $f(x) = f(y)$ then $x \prc y$, hence by duality $x \# y$.
  \end{enumerate}
  Moreover, if $\ES$, $\ES'$ have global precedence, then
  \begin{enumerate}
    \setcounter{enumi}{2}    
  \item if $x \prc y$ and $\neg (y \prc x)$ then $f(x) \prc f(y)$;
  \end{enumerate}
\end{alemma}

\begin{proof}
  Let $x, y \in E$
  \begin{enumerate}
  \item Assume $f(x) \prc f(y)$. Let $C \in \conf{\ES}$ be a
    configuration such that $x, y \in C$. Then $f(x), f(y) \in f(C)$
    and $C \in \conf{\ES'}$. Since $f(x) \prc f(y)$ we have that
    $f(x) <_{f(C)} f(y)$ and thus, since $f$ is a morphism, $x <_C
    y$. Since this holds for any configuration, we conclude $x \prc y$.
    
  \item Assume $f(x) = f(y)$. Since $f$ is injective on
    configurations, there cannot be $C \in \conf{\ES}$ such
    that $x, y \in C$. Therefore, trivially $x \prc y$ (and $y \prc x$,
    whence $x \# y$).

  \item If $\ES$, $\ES'$ have global precedence, $f$ is a folding and
    $x \prc y$ and $\neg(y \prc x)$ then $\neg (x \# y)$ and thus
    there is some configuration $C \in \conf{\ES}$ such that
    $x, y \in C$. Since $\ES$ has global precedence, $x \leq_C y$. Now
    $f(x), f(y) \in f(C)$ which is in $\conf{\ES'}$. Therefore
    $f(x) \leq_{f(C)} f(y)$. Again, since $\ES'$ has global
    precedence, $f(x) \prc f(y)$, as desired.

  \end{enumerate}
\end{proof}

\section{Proofs for Section~\ref{se:folding-criteria} (Foldings for Prime and Asymmetric  Event Structures)}

\lepesmorphism*

\begin{proof}
  First observe that {\pes}s have global precedence and $x \prc y$ iff
  $x \leq y$ or $x \# y$.

  Now, assume that $f$ is a morphism. Then property~(1) holds by
  definition. Property~(2) follows from the fact that
  $\causes{x} \in \conf{\ES[P]}$. Hence
  $f(\causes{x}) \in \conf{\ES[P']}$ and
  $f(\causes{x}) \simeq \causes{x}$, which implies
  $f(\causes{x}) = \causes{f(x)}$.

  Concerning condition (3b), observe that from
  \cref{le:morphism-properties}(1), instantiated with the notion of
  $\prc$ for {\pes}s, we get
  \begin{center}
    $f(x) \leq f(y)$ or $f(x) \# f(y)$ implies $x \leq y$ or $x \# y$.
  \end{center}
  In particular, if $f(x) \# f(y)$ then $x \leq y$ or $x \# y$ and,
  since conflict is symmetric, we also have $y \leq x$ or $y \# x$. It
  is now easy to see that only the second possibility $x \# y$ can
  hold true, which is the desired conclusion.
  Property (3a) immediately derives
  from \cref{le:morphism-properties}(2).
  
  Conversely, assume that $f$ satisfies conditions (1)-(3)
  above. Given a configuration $C \in \conf{\ES[P]}$, by conditions
  (2a) and (3b), $f(C)$ is a configuration in $\ES[P']$. By condition
  (3a), $f$ is injective on $C$. This, together with condition (2b),
  implies that $C \simeq f(C)$.
\end{proof}

\prpesfolding*

\begin{proof}
  Let $f : \ES[P] \to \ES[P']$ be a folding. Let us first observe that
  $f$ is surjective. Take $x' \in \ES[P']$. Since
  $\causes{x'} \in \conf{\ES[P']}$, we have that
  $\emptyset \trans{\causes{x'}} \causes{x'}$. Since $f$ is a folding,
  there must be $C \in \conf{\ES[P]}$ such that $f(C)=\causes{x'}$,
  and thus there is $x \in C$ such that $f(x) =x'$, as desired.

  We next show that properties
  (1) and (2) hold.

  \begin{enumerate}
    
  \item 
    We prove the contronominal, namely that if $f(x) \cons y'$
    then there is $y \in \ES[P]$ such that $f(y) = y'$ and
    $x \cons y$. Assume that $f(x) \cons y'$. We distinguish
    two possibilities:
    
    \begin{itemize}
      
    \item If $y' \leq f(x)$ then, by
      \cref{le:pes-morphism}(\ref{le:pes-morphism:2}a), there exists
      $y \leq x$ such that $f(y)=y'$. Hence $x \cons y$, as desired.
      
    \item Assume that, instead, $\neg (y' \leq f(x))$. Therefore, if we
      let $C' = \causes{f(x)} \cup\causes{y'}$ and
      $X' = C' \setminus \causes{f(x)}$
      \begin{equation}
        \label{eq:pes-folding:1}
        \causes{f(x)} \trans{X'} C'
      \end{equation}
      By \cref{le:pes-morphism}(\ref{le:pes-morphism:2}), we have
      that $f(\causes{x}) = \causes{f(x)}$. Therefore, since $f$ is a
      folding, there must be a transition $\causes{x} \trans{X} C$
      with $f(C) = C'$. This means that there exists $y \in C$ such
      that $f(y) \in C'$ and, since $x \in C$, necessarily
      $x \cons y$, as desired.
    \end{itemize}

  \item Assume that $\setcons{(X \cup \{ x\})}$, $\setcons{(Y \cup \{ y\})}$,
    $\setcons{(X \cup Y)}$ and $f(x) = f(y)$.
    Define $C = \causes{X \cup Y} \in \conf{\ES[P]}$. We distinguish two cases.

    \begin{itemize}
    \item If $x \in C$ then we can simply take $z=x$, since clearly
      $\setcons{(X \cup Y \cup \{ x\})}$.
      
    \item Assume now that $x \notin C$. Clearly $f(x) \notin
      f(C)$. Moreover, $\setcons{(f(C) \cup \{f(x)\})}$. In fact, by
      \cref{le:pes-morphism}(\ref{le:pes-morphism:3}), if for
      some $w \in C$ it were $f(w) \# f(x) = f(y)$ we would have
      $w \# x$ and $w \# y$, contradicting either
      $\setcons{(X \cup \{x\})}$ or $\setcons{(Y \cup \{y\})}$.
      
      Therefore $f(C) \trans{X'} f(C) \cup \causes{f(x)}$ with
      $X' = f\causes{f(x)} \setminus f(C)$. Since $f$ is a folding,
      this implies that $C \trans{X} D$ with
      $f(D) = f(C) \cup \causes{f(x)}$ and
      $D \simeq f(C) \cup \causes{f(x)}$. Therefore there exists
      $z \in D$ such that $f(z)=f(x)$. Since $X \cup Y \subseteq D$,
      we have that $\setcons{(X \cup Y \cup \{z\})}$, as desired.
    \end{itemize}      

    \smallskip

    For the converse implication, assume that $f$ is a surjective morphisms
    satisfying conditions (1) and (2). We have to prove that it is a
    folding.

    Let $C_1 \in \conf{\ES[P]}$ and assume that $f(C_1) \trans{x'} C_2'$.
    If $C_1=\emptyset$, take any $x \in \ES[P]$ such that $f(x)=x'$,
    which exists by surjectivity. By
    \cref{le:pes-morphism}(\ref{le:pes-morphism:2}b) we have
    $f(\causes{x}) = \causes{x'} = \{x'\}$, and thus
    $\causes{x}=\{x\}$. This means that
    $C_1=\emptyset \trans{x} \{x\}$, and we conclude.
    
    Otherwise, if $C_1 \neq \emptyset$, since for all $y \in C_1$ it
    holds that $f(y) \cons x'$, by condition (1), there exists some
    element $x_y \in \ES[P]$ such that $x_y \cons y$ and $f(x_y) =
    x$. Note that necessarily $\neg (x_y \leq y)$, otherwise, by
    \cref{le:pes-morphism}(\ref{le:pes-morphism:2}b) we would have
    $x' = f(x_y) \leq f(y)$, which is not the case.

    Since $C_1$ is finite and consistent, an inductive argument based
    on condition (2), allows to derive the existence of $x$ such that
    $f(x) = x'$ and $\setcons{(C_1 \cup \{x\})}$. Moreover, as argued
    above for the $x_y$s, it is not the case that $x \leq y$ for some
    $y \in C_1$. Therefore there is a transition
    \begin{center}
      $C_1 \trans{X} C_1 \cup \causes{x}$
    \end{center}
    where $X = \causes{x} \setminus C_1$.

    We argue that $X = \{x\}$ and thus we conclude. In fact, assume
    that there is some $z \in X \setminus \{x\}$. Since $f$ is a
    morphism $f(z) \leq f(x) = x'$. Now, since there is the transition
    $f(C_1) \trans{x'}$, all causes of $x'$ must be in $f(C_1)$. Note
    that, since $f$ is a morphism, by
    \cref{le:pes-morphism}(\ref{le:pes-morphism:2}), we have
    $\causes{x'} = \causes{f(x)} = f(\causes{x})$. Therefore, there
    must exist $z_1 \in C_1$ such that $f(z_1) = f(z)$. However, since
    $z, z_1 \in C_1 \cup (\causes{x} \setminus \{x\})$ which is a
    configuration in $\conf{\ES[P]}$, and $f$ is injective on
    configurations, we get $z=z_1 \in C_1$, contradicting the
    hypothesis.
  \end{enumerate}
  
\end{proof}

Given a {\pes} $\ES[P]$ and an event $x \in \ES[P]$ let us define $\scauses{x} = \causes{x} \setminus\{x\}$, $\consequences{x} = \{ y \mid y \in \ES[P]\ \land\ x < y\}$, and $\cset{x} = \{ y \mid y \in \ES[P]\ \land\ \neg(x \leq y\ \lor\ y \leq x\ \lor\ x\#y )\}$.

\begin{adefinition}[abstraction homomorphisms~\cite{Cas:phd}]
  \label{de:abstraction}
  Let $\ES[P]$, $\ES[P']$ be {\pes}s. An \emph{abstraction morphism} is a function $f : \ES[P] \to \ES[P']$ such that for all
  $x, y \in \ES[P]$
  \begin{enumerate}
  \item \label{de:abstraction:1}
    $\lambda'(f(x)) = \lambda(x)$;
    
  \item \label{de:abstraction:2}
    $f(\scauses{x}) = \scauses{f(x)}$;
    
  \item \label{de:abstraction:3}
    $f(\consequences{x}) = \consequences{f(x)}$;
    
  \item \label{de:abstraction:4}
    $f(\cset{x}) = \cset{f(x)}$
  \end{enumerate}
\end{adefinition}

\begin{alemma}[foldings vs abstraction homomorphisms]
  \label{le:abs-fold}
  Let $\ES[P]$, $\ES[P']$ be {\pes} and let $f : \ES[P] \to \ES[P']$
  be a function. Then $f$ is an abstraction morphism iff $f$ is a
  {\pes} morphism additionally satisying condition
  (\ref{pr:pes-folding:1}) of Lemma~\ref{pr:pes-folding}.
\end{alemma}

\begin{proof}
  Let $f$ be an abstraction homomorphism. We first prove conditions
  (\ref{le:pes-morphism:1})-(\ref{le:pes-morphism:3}) of
  Lemma~\ref{le:pes-morphism}. The first condition is already in
  Definition~\ref{de:abstraction}. Condition~(\ref{le:pes-morphism:2}),
  is immediately implied by
  Definition~\ref{de:abstraction}(\ref{de:abstraction:2})
  Concerning
  condition~(\ref{le:pes-morphism:3}), let $x, y \in \ES[P]$ such that
  $f(x)=f(y)$ and $x \neq y$. Observe that we cannot have $x < y$,
  otherwise by
  Definition~\ref{de:abstraction}(\ref{de:abstraction:2}), we would
  have $f(x) < f(y)$. Dually, it cannot be $y<x$.  Moreover, it cannot
  be $x \in \cset{y}$, otherwise
  Definition~\ref{de:abstraction}(\ref{de:abstraction:4}) would be
  violated. Therefore, necessarily $x \# y$.
  The validity of condition~(\ref{le:pes-morphism:3}b) is proved
  analogously.
  
  We finally show that $f$ satisfies also condition
  (\ref{pr:pes-folding:1}) of Lemma~\ref{pr:pes-folding}. Let
  $x \in \ES[P]$, $y \in \ES[P']$ such that $\neg (f(x) \# y')$ and we
  show that $\neg (x \# y)$ for some $y \in \ES[P]$ such that
  $f(y)=y'$. We distinguish various possibilities:
  \begin{itemize}
  \item If $f(x) = y'$,  we simply take  $y=x$.
    
  \item If $y' < f(x)$, by
    Definition~\ref{de:abstraction}(\ref{de:abstraction:2}) there
    exists $y \in \ES[P]$ with $y < x$ such that $f(y)=y'$, and we conclude.
    
  \item If $f(x) < y'$, by
    Definition~\ref{de:abstraction}(\ref{de:abstraction:3}) there
    exists $y \in \ES[P]$ with $x < y$ such that $f(y)=y'$, and we
    conclude.

  \item If none of the above holds, necessarily $y' \in \cset{f(x)}x$,
    and thus by
    Definition~\ref{de:abstraction}(\ref{de:abstraction:4}) there
    exists $y \in \ES[P]$ with $y \in \cset{x}$ such that $f(y)=y'$,
    and we conclude.
  \end{itemize}

  \bigskip

  Conversely, let $f$ be a {\pes} morphism additionally satisying
  condition (\ref{pr:pes-folding:1}) of Lemma~\ref{pr:pes-folding}. We
  prove that conditions
  (\ref{de:abstraction:1})-(\ref{de:abstraction:4}) of
  Definition~\ref{de:abstraction} hold. As above, the first conditions
  is already in Lemma~\ref{le:pes-morphism}. The second condition,
  namely $f(\scauses{x}) = \scauses{f(x)}$ immediately follows from
  Lemma~\ref{le:pes-morphism}(\ref{le:pes-morphism:2}), i.e,
  $f(\causes{x}) = \causes{f(x)}$. In fact, we only need to observe
  that for all $y < x$, $f(y) \neq f(x)$, otherwise, by
  Lemma~\ref{le:pes-morphism}(\ref{le:pes-morphism:3}a) we would have
  $x \# y$.

  Concerning (\ref{de:abstraction:3}), i.e., for $x \in \ES[P]$,
  $f(\consequences{x}) = \consequences{f(x)}$ let us prove separately
  the two inclusions.
  \begin{itemize}
  \item ($\subseteq$) Let $y' \in f(\consequences{x})$, i.e.,
    $y' = f(y)$ for some $y \in \consequences{x}$. Since $x < y$, by
    Lemma~\ref{le:pes-morphism}(\ref{le:pes-morphism:2}b),
    $f(x) < f(y)$ and thus $y' = f(y) \in \consequences{f(x)}$ , as
    desired.
  
  \item ($\supseteq$) Let $y' \in \consequences{f(x)}$, i.e.,
    $f(x) < y'$. Then, for all $y \in f^{-1}(y')$, since
    $f(x) < y' = f(y)$, by
    Lemma~\ref{le:pes-morphism}(\ref{le:pes-morphism:2}a), there is
    $z < y$ such that $f(z) = f(x)$. Hence either $z=x$ and thus
    $x < y$ or $z \neq x$, hence, by
    Lemma~\ref{le:pes-morphism}(\ref{le:pes-morphism:3}a), $x \# z$
    and thus $x \# y$.

    It cannot be that $x \#^\forall f^{-1}(y')$ , otherwise, by
    Lemma~\ref{pr:pes-folding}(\ref{pr:pes-folding:1}), we would have
    $x \# y$, which is not the case. Therefore there must exists
    $y \in f^{-1}(y')$ such that $x < y$. Therefore
    $y' = f(y) \in f(\consequences{x})$.
  \end{itemize}

  Let us finally prove condition (\ref{de:abstraction:4}), i.e., for
  $x \in \ES[P]$, $f(\cset{x}) = \cset{f(x)}$. Again, we prove
  separately the two inclusions.
    \begin{itemize}
    \item ($\subseteq$) Let $y' \in f(\cset{x})$, i.e., $y' = f(y)$
      for some $y \in \cset{x}$. By
      Lemma~\ref{le:pes-morphism}(\ref{le:pes-morphism:2}b) and
      Lemma~\ref{le:pes-morphism}(\ref{le:pes-morphism:3}b), it must
      be $y' = f(y) \in \cset{f(x)}$, as desired.

    \item ($\supseteq$) Let $y' \in \cset{f(x)}$. Since
      $\neg (f(x) \# y')$, by
      Lemma~\ref{pr:pes-folding}(\ref{pr:pes-folding:1}), we deduce
      that $\neg (x \#^\forall f^{-1}(y'))$.
      Take any $y \in f^{-1}(y')$ such that $\neg (x \# y)$.
      Now observe that it cannot be $x < y$ or $y < x$, otherwise, by
      Lemma~\ref{le:pes-morphism}(\ref{le:pes-morphism:2}b) $f(x)$
      and $y'=f(y)$ would be ordered in the same way, contradicting
      $y' \in \cset{f(x)}$.
      It cannot be $x=y$ either, otherwise $y'=f(y) = f(x)$, again
      contradicting $y' \in \cset{f(x)}$.

      Therefore, $y \in \cset{x}$ and thus $y' = f(y) \in f(\cset{x})$, as desired.
    \end{itemize}

\end{proof}

For instance, consider the {\pes}s $\ES[P_7]$ and $\ES[P_8]$ in
Fig.~\ref{fi:abs-fold}. It can be seen that obvious function
$f_{78} : \ES[P_7] \to \ES[P_8]$ is an abstraction homomorphism but
not a folding. Indeed, consider the configuration $\{ b_0, a_1\}$. Then the step $f_{78}(\{ b_0, a_1\}) \trans{c_{01}} \{ b_{01}, a_{01}, c_{01}\}$ cannot be simulated by $\{ b_0, a_1\}$.

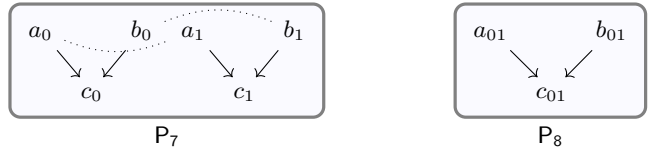
\begin{figure}
  \begin{center}
  \setlength{\tabcolsep}{24pt}
  \begin{tabular}{ccc}
    \begin{tikzcd}[boxedcd, column sep=1mm, row sep=4mm]
      a_0
      \arrow[dr] \arrow[conflict, bend right=22, rrr]
      & &       
      b_0 \arrow[dl]  \arrow[conflict, bend left=22, rrr]
      &
      a_1
      \arrow[dr]
      & &       
      b_1 \arrow[dl]
      \\
      & c_0 & & & c_1
    \end{tikzcd}
    &
    \begin{tikzcd}[boxedcd, column sep=1mm, row sep=4mm]
       a_{01}
      \arrow[dr] %
      & &       
      b_{01} \arrow[dl]\\
      & c_{01}
    \end{tikzcd}
    \\
    $\ES[P_7]$ & $\ES[P_8]$
  \end{tabular}
  \end{center}
  \caption{Abstraction homomorphisms vs folding morphisms.}
  \label{fi:abs-fold}
\end{figure}

\cofoldeqpes*

\begin{proof}
  Let $\ES[P]$ be a {\pes}s and let $\equiv$ be a folding
  equivalence. This means that there exists a folding
  $f : \ES[P] \to \ES[P]'$ such that $\equiv$ and $\equiv_f$ coincide. By
  \cref{le:fold-is-equivalence} we know that
  $\quotient{\ES[P]}{\equiv_f}$ is isomorphic to $\ES[P']$. Therefore
  using \cref{le:pes-morphism} and \cref{pr:pes-folding} we
  immediately get the validities of properties (1)-(5).

  Conversely, assume that $\equiv$ satisfies properties (1)-(5)
  above. Define a {\pes} $\ES[P']$ as follows.
  \begin{itemize}
  \item $E' = \quotient{E}{\equiv}$;
  \item $\eqclass{x}{\equiv} \leq' \eqclass{y}{\equiv}$ if
    $\eqclass{x}{\equiv} \leq^\exists \eqclass{y}{\equiv}$;
  \item $\eqclass{x}{\equiv} \#' \eqclass{y}{\equiv}$ if
    $\eqclass{x}{\equiv} \#^\forall \eqclass{y}{\equiv}$
  \item $\lambda'(\eqclass{x}{\equiv}) = \lambda(x)$.
  \end{itemize}
  Observe that $\ES[P']$ is a well-defined {\pes}. A simple key
  observation is that
  \begin{equation}
    \label{eq:trans}
    \eqclass{x}{\equiv} \leq' \eqclass{y}{\equiv} \leq'
    \eqclass{z}{\equiv} \quad \Rightarrow \quad
    \exists x' \in \eqclass{x}{\equiv}.\
    y' \in \eqclass{y}{\equiv}.\
    z' \in \eqclass{z}{\equiv}.\  x \leq y \leq z
  \end{equation}
  In fact, since $\eqclass{y}{\equiv} \leq' \eqclass{z}{\equiv}$, by
  definition we have the existence of $y' \in \eqclass{y}{\equiv}$ and
  $z' \in \eqclass{z}{\equiv}$ such that $y' \leq z'$.
  Moreover, since $\eqclass{x}{\equiv} \leq' \eqclass{y}{\equiv}$, by
  definition we have the existence of $x'' \in \eqclass{x}{\equiv}$ and
  $y'' \in \eqclass{y}{\equiv}$ such that $x'' \leq y''$.
  Since $y' \equiv y''$, by condition (2),
  $\eqclass{y'}{\equiv} = \eqclass{y''}{\equiv}$. Hence from
  $x'' \leq y''$ we deduce the existence of $x' \leq y'$ with
  $x' \in \eqclass{x}{\equiv}$ as desired.

  Using (\ref{eq:trans}), we can immediately inherit the partial order
  properties of $\leq'$ and irreflexivity and hereditarity of $\#'$
  from the analogous properties of $\#$.

  If we define a function $f : \ES[P] \to \ES[P']$ as
  $f(x) = \eqclass{x}{\equiv}$, it is now easy to show that it
  satisfies properties
  (\ref{le:pes-morphism:1})-(\ref{le:pes-morphism:3}) in
  \cref{le:pes-morphism}, and 
  (\ref{pr:pes-folding:1}),(\ref{pr:pes-folding:2}) in
  \cref{pr:pes-folding}, hence it is a folding and we conclude.
\end{proof}

\lefoldjoinpes*

\begin{proof}
  The result can be proved by using the fact that, by
  \cref{le:factor-fold}, $\cat{PES_f}$ is a coreflective category of
  $\cat{ES_f}$, hence it is closed under pushout as proved
  in~\cite[Corollary 1]{HS:CS}.
  
  Explicitly, the fact that $\Pr(g') : \ES[P'] \to \pr(\ES''')$ and
  $\pr(g'') : \ES[P''] \to \pr(\ES''')$ are foldings derive from
  \cref{pr:fold-es-pes}.
  Now observe that, since In order to show that this actually provide
  a pushout in $\cat{PES}$, consider two morphisms $g_1'$ and $g_2'$
  as in the diagram below, such that
  $g_1' \circ f' = g_1'' \circ f''$:
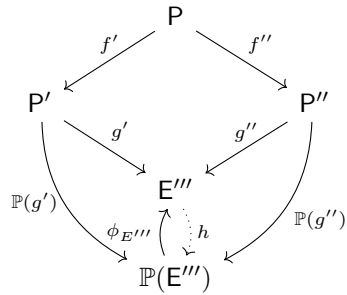
  \begin{center}
    \begin{tikzcd}%
      &
      \ES[P] \arrow[dl, "f'" above] \arrow[dr, "f''"]
      &
      \\
      \ES[P']
      \arrow[dr, "g'"]
      \arrow[ddr, bend right, "\pr(g')" left]
      &
      {}
      &
      \ES[P'']  \arrow[dl, "g''" above]
      \arrow[ddl, bend left, "\pr(g'')"]\\
      &
      \ES''' \arrow[d, bend left=25, dotted, "h" right]
      &\\
      &
      \pr(\ES''') \arrow[u, bend left=25, "\phi_{E'''}" left]
      &
    \end{tikzcd}
  \end{center}

  Since $\ES'''$ is a pushout and
  $\pr(g') \circ f' = \pr(g'') \circ f''$, there is a unique morphism
  $h : \ES''' \to \pr(\ES''')$, making the diagram commute. Now,
  observe that $\phi_{\ES'''} \circ h : \ES''' \to \ES'''$ can be used
  in the diagram below as mediating morphisms:
  
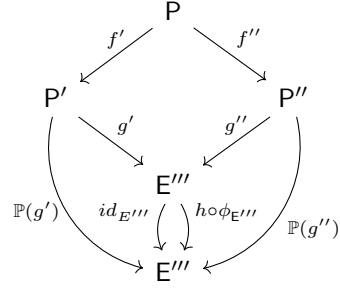
\begin{center}
    \begin{tikzcd}%
      &
      \ES[P] \arrow[dl, "f'" above] \arrow[dr, "f''"]
      &
      \\
      \ES[P']
      \arrow[dr, "g'"]
      \arrow[ddr, bend right=45, "\pr(g')" left]
      &
      {}
      &
      \ES[P'']  \arrow[dl, "g''" above]
      \arrow[ddl, bend left=45, "\pr(g'')"]\\
      &
      \ES'''
      \arrow[d, bend right=25, "id_{E'''}" above left]
      \arrow[d, bend left=25, "h \circ \phi_{\ES'''}" above right]
      &\\
      &
      \ES'''  
      &
    \end{tikzcd}
  \end{center}

  Now, also the identity works as mediating morphisms we deduce that
  $h \circ \phi_{\ES'''} = id_{\ES'''}$, which implies that
  $\phi_{\ES'''}$ is injective. Since it is a folding, it is also
  surjective, and therefore it is an isomorphism, as desired.
\end{proof}

\leaesmorphism*

\begin{proof}
  Let $f : \ES[A] \to \ES[A']$ be a morphism. Just observe that
  {\pes}s have global precedence and $x \prc y$ iff $x \ac y$.
  Condition (\ref{le:aes-morphism:1}) is obviously true. Property (2)
  follows by observing that, for all $x \in \ES[A]$, since
  $\causes{x} \in \conf{\ES[A]}$ and $f$ is a morphism, then
  $f(\causes{x}) \in \conf{\ES[A]}$. Since configurations are causally
  closed we deduce that $\causes{f(x)} \subseteq f(\causes{x})$. The
  validity of properties (\ref{le:aes-morphism:3}) and
  (\ref{le:aes-morphism:4} is given directly by
  \cref{le:morphism-properties}.

  Conversely, assume that $f : \ES[A] \to \ES[A']$ enjoys properties
  (\ref{le:aes-morphism:1})-(\ref{le:aes-morphism:4}). Let
  $C \in \conf{\ES[A]}$ be a configuration. Function $f$ is injective
  on $C$ since, otherwise, if there are $x, y \in C$ such that
  $f(x)=f(y)$ and $x\neq y$, we would get $x \ac y \ac x$,
  contradicting acyclicity of $\ac$ in $C$. Observe that $f(C)$ is a
  configuration. In fact, $\ac$ is acyclic in $f(C)$ since $C$ is and,
  by (\ref{le:aes-morphism:3}a), cycles are reflected by $f$. In
  addition, $f(C)$ is causally closed by (\ref{le:aes-morphism:2}),
  since $C$ is. Finally, note that $C \simeq f(C)$. In fact, for all
  $x, y \in C$, if $x \ac y$, since $\neg (y \ac x)$, by
  (\ref{le:aes-morphism:3}b), we get $f(x) \ac f(y)$. Conversely, if
  $f(x) \ac f(y)$ then $x \ac y$, by (\ref{le:aes-morphism:3}a).
\end{proof}

\praesfolding*

\begin{proof}
  Let $f : \ES[A] \to \ES[A']$ be a folding. Surjectivity of $f$ can
  be proved exactly as in Proposition~\ref{pr:pes-folding}. We show
  that properties (1)-(3) hold.

  \begin{enumerate}
    
  \item 
    We prove the contronominal, namely that if
    $\neg (y' \ac^\exists f(\causes{x}))$ then there is $y \in \ES[A]$
    such that $f(y) = y'$ and $\neg (y \ac x)$. Let
    $H = \causes{x} \in \conf{\ES[A]}$ and assume that
    $\neg (y' \ac^\exists f(H))$. Since $f$ is a morphism
    $H' = f(H) \in \Hist{f(x)}$. Observe that we can safely assume
    that $y' \not\in H'$. In fact, otherwise, since
    $\neg (y' \ac^\exists H')$, the only possibility would be
    $y'=f(x)$ and thus we could take $y=x$ since $\neg (x \ac x)$, as
    desired.
    Using the fact that $\neg (y' \ac^\exists H')$ and
    $y \notin H'$, if we let $C'= H' \cup \causes{y'}$ and
    $Y' = C' \setminus H'$
    \begin{equation}
      \label{eq:aes-folding:2}
      H' \trans{Y'} C'
    \end{equation}
    Therefore, since $f$ is a folding, there must be a transition
    $H \trans{X} C$ with $f(C) = C'$. This means that there exists
    $y \in X$ such that $f(y) = y'$ and since $H=\causes{x}$,
    necessarily $\neg (y \ac x)$, as desired.

  \item Assume that $x \notin X$, $y \notin Y$
    $\neg (x \ac^\exists X)$, $\neg (y \ac^\exists Y )$,
    $\setcons{(X \cup Y)}$ and $f(x) = f(y)$.
    Define $C = \causes{X \cup Y} \in \conf{\ES[A]}$. We show that
    $x \not\in C$. In fact, $x \notin \causes{X}$ since $x \notin X$
    and $\neg (x \ac^\exists X)$, and, for analogous reasons,
    $y \notin \causes{Y}$.
    Now, if $x=y$ we are done. Otherwise, we can prove that
    $x \not\in \causes{Y}$ and conclude. In fact, assume by
    contradiction that $x \in \causes{Y}$, i.e., $x \leq w$ for some
    $w \in Y$. Since $f(x)=f(y)$ and $x \neq y$, we deduce, by
    \cref{le:aes-morphism}(\ref{le:aes-morphism:4}), that $y \ac
    x$. Recalling $x \leq w$, by inheritance of asymmetric conflict,
    we get $y \ac^\exists Y$, contradicting the hypotheses.

    Since $x \notin C$, we have $f(x) \notin f(C)$. Moreover, if we
    let $y' = f(x)=f(y)$, we have $\neg (y' \ac^\exists
    f(C))$. Otherwise, by
    \cref{le:aes-morphism}(\ref{le:aes-morphism:3}a), we would deduce
    $x \ac^\exists X$ or $y \ac^\exists Y$, contradicting the
    hypotheses.
    
    Therefore $f(C) \trans{X'} f(C) \cup \causes{f(x)}$ with
    $X' = f\causes{f(x)} \setminus f(C)$. Since $f$ is a folding, this
    implies that $C \trans{X} D$ with $f(D) = f(C) \cup \causes{f(x)}$
    and $D \simeq f(C) \cup \causes{f(x)}$. Therefore there exists
    $z \in D$ such that $f(z)=f(x)$. Therefore
    $\neg (z \ac^\exists C)$. Hence, recalling
    $C = \causes{X} \cup \causes{Y}$, we have
    $\neg (z \ac^\exists X \cup Y)$, as desired.

  \item Take $H \in \Hist{x}$ with $\neg (H \ac^\exists X)$ and
    $H_1 \sqsubsetneq H$ such that
    $f(H_1) \cup \{f(x)\} \in \Hist{f(x)}$, hence
    $f(H_1) \trans{f(x)} f(H_1) \cup \{f(x)\}$.
    Consider $C = H_1 \cup \causes{X}$.  Since
    $H_1 \cup\{x\} \subseteq H$ and  $\neg ( H \ac^\exists X)$, we have
    $\neg (H_1 \cup\{x\} \ac^\exists \causes{X})$ and thus, by
    \cref{le:aes-morphism}(\ref{le:aes-morphism:3}a),
    $\neg (f(H_1 \cup \{x\}) \ac^\exists f(\causes{X})$. Therefore
    $f(H_1 \cup \causes{X}) = f(H_1) \cup f(\causes{X}) \trans{f(x)}
    C_1'$, and since $f$ is a folding
    $H_1 \cup \causes{X} \trans{x_1} C_1$, with $f(x_1) = f(x)$ and
    clearly (given that the transition exists, $x_1 \ac^\exists X$, as
    desired.
  \end{enumerate}

  \smallskip
  
  For the converse implication, assume that $f$ is a surjective morphism
  satisfying conditions (1)-(3). We have to prove that it is a
  folding.

  Let $C_1 \in \conf{\ES[A]}$ and assume that $f(C_1) \trans{x'} C_2'$.
  When $C_1=\emptyset$ we argue as in Proposition~\ref{pr:pes-folding}.
  Otherwise, if $C_1 \neq \emptyset$, for all $y \in C_1$ it holds
  $\causes{y} \subseteq C_1$ and thus
  $\neg (x' \ac^\exists{f(\causes{y})}$. Thus, by condition (1), there
  exists some element $x_y \in \ES[A]$ such that $f(x_y) = x'$ and
  $\neg (x_y \ac y)$.
  Note that necessarily
  $x_y \neq y$,
  
  Since $C_1$ is finite and consistent, an inductive argument based
  on condition (2), allows to derive the existence of $x$ such that
  $f(x) = x'$ and $\neg (x \ac^\exists C_1)$.
  Therefore there is a transition
  \begin{center}
    $C_1 \trans{X} C_2$
  \end{center}
  where $C_2 = C_1 \cup \causes{x}$ and $X = \causes{x} \setminus C_1$.
  
  Let $H = \hist{C_2}{x}$. By definition of history, if
  $\neg (H \ac^\exists C_2 \setminus H)$. Let
  $H_1' = \hist{f(C_1)}{x'} \setminus \{x'\}$ and let $H_1$ its
  $f$-counterimage in $C_1$. We have $H_1 \sqsubseteq H$,
  $x'=f(x) \notin f(H_1)$ and $f(H_1) \cup \{f(x)\} \in
  \Hist{f(x)}$. Then, by condition (3), there exists $x_1$ such that
  $H_1 \cup \{ x_1 \} \in \Hist{x_1}$ and
  $\neg (x_1 \ac^\exists C_2 \setminus H)$, hence
  $\neg (x_1 \ac^\exists C_1 \setminus H_1)$. This implies
  $C_1 \trans{x_1} C_1 \cup \{ x_1 \}$, as desired.
\end{proof}

\end{document}